\documentclass[twoside]{article}

\usepackage[preprint]{aistats2027}

\usepackage[utf8]{inputenc} 
\usepackage[T1]{fontenc}    
\usepackage{hyperref}       
\usepackage{url}            
\hypersetup{hidelinks}
\usepackage[round]{natbib}
\renewcommand{\bibname}{References}
\renewcommand{\bibsection}{\subsubsection*{\bibname}}
\usepackage{booktabs}       
\usepackage{amsfonts}       
\usepackage{nicefrac}       
\usepackage{microtype}      
\usepackage{xcolor}         

\usepackage{graphicx}
\usepackage{longtable}
\usepackage{subcaption}
\usepackage{bbm}
\usepackage{amsmath}
\usepackage{amssymb}
\usepackage{mathtools}
\usepackage{amsthm}
\usepackage{algorithm}
\usepackage{algorithmic}
\usepackage[capitalize,noabbrev]{cleveref}
\usepackage{newunicodechar}

\newunicodechar{—}{---}
\newunicodechar{–}{--}
\newunicodechar{’}{'}

\theoremstyle{plain}
\newtheorem{theorem}{Theorem}[section]
\newtheorem{proposition}[theorem]{Proposition}
\newtheorem{lemma}[theorem]{Lemma}
\newtheorem{corollary}[theorem]{Corollary}
\theoremstyle{definition}

\newtheorem{assumption}[theorem]{Assumption}
\theoremstyle{remark}

\begin{document}

\runningauthor{Renjie Cao, Zhuoxin Yan, Zichen Wang, Xinyan Su, Zhiheng Zhang}

\twocolumn[
\aistatstitle{Causal Influence Maximization with Steady-State Guarantees}

\aistatsauthor{Renjie Cao$^{*}$ \And Zhuoxin Yan$^{*}$ \And Zichen Wang}
\aistatsaddress{Boston College \\ \\ {\small\texttt{caoren@bc.edu}}
  \And Shanghai University of \\ Finance and Economics \\ {\small\texttt{2024111763@stu.sufe.edu.cn}}
  \And University of Illinois \\ at Urbana-Champaign \\ {\small\texttt{zichenw6@illinois.edu}}}

\aistatsauthor{Xinyan Su \And Zhiheng Zhang$^{\dagger}$}
\aistatsaddress{University of Chinese \\ Academy of Sciences \\ {\small\texttt{studysxy@163.com}}
  \And Shanghai University of \\ Finance and Economics \\ {\small\texttt{zhangzhiheng@mail.shufe.edu.cn}}}
]

{\renewcommand{\thefootnote}{}\footnotetext{$^{*}$Equal contribution. \quad $^{\dagger}$Corresponding author.}}

\begin{abstract}
Classical influence maximization optimizes expected reach, which can be misaligned with welfare when exposures have heterogeneous, saturating, or adverse effects. We study seed selection for steady-state causal welfare under network interference. Under our exposure-separable outcome model and weak independent live-edge propagation, a deterministic surrogate built from expected activations and exposure counts approximates welfare with uniform $O(\varepsilon^2)$ error for a fixed graph and seed budget, where $\varepsilon$ bounds edge activation probabilities. Building on this reduction, our two-stage method, CIM, learns shape-constrained exposure-response functions from logged diffusion-outcome data and estimates expected activations and exposures by Monte Carlo simulation. For unregularized fits on a design that covers the evaluated exposures, and under noise and simulation conditions, we bound the welfare-estimation error uniformly over seed sets, separating structural, statistical, and simulation errors; with an approximation guarantee for the optimizer, the bound transfers to the selected seeds. Experiments on real network topologies with simulated outcomes and a synthetic benchmark show that CIM attains the highest mean welfare among the compared methods.
\end{abstract}

\section{Introduction}
\label{sec:intro}

Interventions in networked populations can affect people beyond their initial recipients. Such effects have been studied in online social networks \citep{bond2012mobilization}, public health \citep{kim2015targeting}, and online marketplaces \citep{johari2022experimental,holtz2024reducing}. Classical influence maximization selects the initial participants to maximize expected reach, measured by the number of individuals who eventually become active \citep{kempe2003maximizing}. The value of this reach depends on the application: treatment responses and exposure effects determine the outcomes produced by a cascade \citep{viviano2025policy,su2023cauim}.

Consider a platform that grants a ``trusted source'' badge to a few seed users to improve information quality. A high-degree seed may generate a large cascade concentrated among users who already receive similar information. Another seed may reach fewer users while bringing useful information to a less exposed community. We study seed selection for the welfare produced by such interventions after diffusion stabilizes. The seed set determines the initial intervention, propagation generates the final treatment state, and each individual's outcome depends on their own activation and on how many of their contacts are active. Classical influence maximization is recovered when each individual's outcome equals their final activation status.

Evaluating this welfare is difficult because the final treatment state depends on correlated propagation events. Nonlinear outcome responses must be averaged over the resulting exposure distribution, which can be complex even when the diffusion model is known. Our structural reduction approximates this expectation using expected final activations and expected exposure counts. Under the exposure-separable outcome assumptions, the approximation error is controlled by response curvature and the frequency of multiple non-seed exposures. Hence the outcome model (through curvature) and the propagation process (through repeated non-seed exposures) jointly determine the accuracy of the reduction.

The reduction leads to a two-stage Causal Influence Maximization (CIM) procedure. We first learn exposure-response functions from logged diffusion-outcome observations using shape-constrained regression \citep{groeneboom2014nonparametric,guntuboyina2018nonparametric}. Monte Carlo simulation of the diffusion model \citep{owen2013montecarlo} then estimates expected activations and exposures for candidate seed sets. We combine these quantities to evaluate welfare and select seeds.

Our contributions are threefold. First, we show that replacing random final exposures by their expectations incurs an error controlled by response curvature, and that this error is $O(\varepsilon^2)$ under weak independent live-edge propagation, uniformly over seed sets for a fixed graph and seed budget (Theorem~\ref{thm:exposure-compression}, Corollary~\ref{cor:welfare-approx}). Second, for unregularized fits on a covered design and under the noise and simulation conditions of Section~\ref{subsec:estimation-fixed}, we bound the welfare-estimation error of CIM uniformly over seed sets, separating structural, statistical, and simulation errors (Theorem~\ref{thm:fixedS-finite}); an approximation guarantee for the optimizer transfers this bound to the selected seeds (Corollary~\ref{thm:end2end}). Third, on five benchmarks, CIM attains the highest mean welfare among the compared methods (Section~\ref{sec:experiments}).

\section{Related Work}

Appendix~\ref{app:literature} gives a more detailed review.

\paragraph{Influence maximization.}
Classical influence maximization selects a seed set to maximize expected cascade size under diffusion models such as Independent Cascade (IC) and Linear Threshold \citep{kempe2003maximizing}. Expected spread is monotone and submodular in these models, which supports the standard greedy guarantee \citep{nemhauser1978analysis}. Later work improves the efficiency of influence estimation and seed selection \citep{borgs2014nearlyoptimal,tang2014tim,tang2015martingale}.

\paragraph{Causal inference under interference.}
Potential-outcome methods with interference allow an individual's outcome to depend on the treatments of others \citep{hudgens2008toward,aronow2017estimating,eckles2017design,athey2018exact}. Exposure mappings reduce the full assignment vector to a lower-dimensional quantity that can be estimated from data. Graph cluster randomization controls the resulting exposure probabilities \citep{ugander2013graph}. Most of this literature studies a one-shot assignment; related work also considers approximate neighborhood interference \citep{leung2022approximate} and sequential experiments \citep{zhang2025online}. In our setting, the relevant assignment is the random limit of a diffusion process.

\paragraph{Causal objectives for network interventions.}
Recent work studies policy targeting under interference and heterogeneous values in influence maximization \citep{viviano2025policy,su2023cauim}; both lines of work show that reach can be an incomplete objective. Network-targeting experiments also evaluate health-product uptake and peer behavior \citep{kim2015targeting,paluck2016conflict}. Closest to our setting, CauIM \citep{su2023cauim} weights each node's activation probability by its individual treatment effect, which may depend on the activation status of its neighbors; it estimates these effects with a neural representation model and analyzes greedy selection through weak submodularity. Our objective instead sums the final outcomes of all individuals, so exposure effects on individuals who never adopt also count, and we model outcomes through monotone concave exposure responses, which yields an explicit approximation bound and finite-sample estimation guarantees.

\section{Problem Formulation}
\label{sec:formulation}

\subsection{Classical influence maximization}

We begin with the seed-selection and diffusion process shared with classical influence maximization.

Let $G=(V,E)$ be a directed graph representing the population and its interaction structure, where
$V=\{1,\dots,n\}$ indexes individuals and $(j,i)\in E$ indicates that individual $j$ can directly influence individual $i$.
Time is indexed by discrete steps $t\in\mathbb{Z}_{\ge 0}$. For each individual $i\in V$ and time $t$, let
$
z_{i,t}\in\{0,1\}
$
denote the \emph{activation state}, where $z_{i,t}=1$ indicates that individual $i$ is active (i.e., treated) at time $t$ and $z_{i,t}=0$ otherwise; outcomes are defined separately in Section~\ref{subsec:welfare}.
The global system state at time $t$ is the vector
$
z_t := (z_{1,t},\dots,z_{n,t})^\top \in \{0,1\}^n .
$ The activation state $z_{i,t}$ may represent adoption of a feature, exposure to information, receipt of an intervention, or any binary treatment whose availability can spread through the network.

\paragraph{Intervention and propagation dynamics.}

At time $t=0$, the experimenter selects a seed set $S\subseteq V$, which defines the initial treatment assignment
\begin{equation}
z_{i,0}(S) = \mathbb{I}(i\in S),
\end{equation}
where $\mathbb{I}(\cdot)$ is the indicator function.
The intervention is subject to a budget constraint $K$, and the feasible policy class is
$
\mathcal{S}_K := \{S\subseteq V : |S|\le K\}.
$

For $t>0$, treatment propagates endogenously according to a stochastic network propagation mechanism defined on $G$.
We keep this mechanism general at the problem-definition stage and impose model-specific conditions only where they are used.

\begin{assumption}[Monotonicity and Convergence]
\label{assump:mono}
The activation process is irreversible: once an individual becomes active, they remain active thereafter.
Formally, if $z_{i,t}=1$ for some $t$, then $z_{i,\tau}=1$ for all $\tau\ge t$.
Consequently, for any seed set $S$, the state sequence $\{z_t(S)\}_{t\ge 0}$ is non-decreasing, and the limit
$
z_\infty(S) := \lim_{t\to\infty} z_t(S)
$
exists almost surely.
\end{assumption}

Assumption~\ref{assump:mono} ensures that the final state $z_\infty(S)$ is well defined. Progressive diffusion models such as Independent Cascade and Linear Threshold satisfy this condition \citep{kempe2003maximizing}. Non-progressive dynamics require a different long-run estimand and lie outside the present analysis.

Classical influence maximization evaluates a seed set through its expected spread,
\begin{equation}
\sigma(S)
:=
\mathbb{E}\!\left[\sum_{i\in V} z_{i,\infty}(S)\right].
\label{eq:spread}
\end{equation}
The classical decision problem maximizes $\sigma(S)$ over $\mathcal{S}_K$. The spread objective assigns the same value to every final activation. A simple example shows how this choice can affect seed selection.

\paragraph{Example: reach and welfare.}
Consider two feasible seed sets with deterministic final cascades. One activates three individuals, each of whom gains $0.1$ from treatment; the other activates two individuals, each of whom gains $1$. Their cascade sizes are three and two, while their total gains are $0.3$ and $2$, respectively. Maximizing reach therefore selects the first seed set, whereas maximizing welfare selects the second. Exposure effects can create a similar disagreement when additional activations repeatedly reach individuals whose responses have already saturated.

\subsection{Steady-state causal welfare}
\label{subsec:welfare}

This example motivates evaluating a seed set through the outcomes of the individuals affected by its cascade. We represent these outcomes using potential outcomes under the final activation state \citep{hudgens2008toward,aronow2017estimating}.

Let
$
Y_i(z), \qquad z\in\{0,1\}^n,
$
denote the mean potential outcome of individual $i$ under global activation state $z$.
This formulation allows $Y_i$ to depend on other individuals' activation states; Assumption~\ref{assump:exposure-separable} specifies the structure used in our analysis.
Because $Y_i(z)$ is a mean outcome, it is a deterministic function of $z$; randomness enters through the final state $z_\infty(S)$ and through observation noise, which may be correlated across individuals (Section~\ref{subsec:estimation-fixed}).

Our target is the \emph{steady-state welfare} induced by a seed set $S$,
defined as
\begin{equation}
F(S)
:=
\mathbb{E}\!\left[
\sum_{i\in V} Y_i\bigl(z_\infty(S)\bigr)
\right],
\label{eq:of}
\end{equation}
where the expectation is taken over the stochastic diffusion process.

\paragraph{Special cases.}
The objective recovers classical influence maximization when $Y_i(z)=z_i$, since $F(S)=\mathbb{E}[\sum_i z_{i,\infty}(S)]=\sigma(S)$. Choosing $Y_i(z)=\alpha_i z_i$ with $\alpha_i\ge 0$ gives expected weighted spread. If propagation stops after the initial assignment, then $z_\infty(S)=z_0(S)$ and the objective becomes welfare under a one-shot treatment assignment with interference. Our formulation goes beyond these cases in two ways: treatment propagates through the network, and outcomes can respond nonlinearly to the exposures it produces.

For a general outcome model, evaluating this objective requires averaging over the distribution of the final treatment state. Our decision problem is
\begin{equation}
S^\star \in \operatorname*{arg\,max}_{S\in\mathcal{S}_K} F(S).
\end{equation}
To make this objective amenable to estimation, we introduce an exposure representation and specify how individual outcomes depend on it.

\subsection{Exposure mapping and outcome model}

Following the exposure-mapping approach to interference \citep{aronow2017estimating}, for each node $i$, fix analyst-specified source sets $N_i^+,N_i^-\subseteq V\setminus\{i\}$.
By default, when no signed exposure information is available, both are taken from the in-neighborhood $\{j:(j,i)\in E\}$; domain knowledge can instead separate beneficial and adverse exposure sources.
For any activation state $z$, define
\[
\begin{aligned}
K_i^+(z)&:=\sum_{j\in N_i^+}z_j,\\
K_i^-(z)&:=\sum_{j\in N_i^-}z_j,\\
\mathcal{E}_i(z)&:=(K_i^+(z),K_i^-(z)).
\end{aligned}
\]
For a seed set $S$, write
\begin{equation}
\label{eq:Kpm-steady}
\begin{aligned}
K_i^\pm(S)&:=K_i^\pm\bigl(z_\infty(S)\bigr),\\
k_i^\pm(S)&:=\mathbb{E}\!\left[K_i^\pm(S)\right],\\
q_i(S)&:=\mathbb{E}\!\left[z_{i,\infty}(S)\right].
\end{aligned}
\end{equation}
The realized counts $K_i^\pm(S)$ are random, whereas $k_i^\pm(S)$ and the final activation probability $q_i(S)$ are expectations that depend on $S$ and are not observed for seed sets that were never deployed.
In the $\ell$-th logged replication, we observe the realized final activation $z_{i\ell}:=z_{i,\infty,\ell}$, the noisy outcome $Y_{i\ell}=Y_i(z_{\infty,\ell})+\xi_{i\ell}$, and the exposure counts
\[
K_{i\ell}^\pm:=K_i^\pm(z_{\infty,\ell});
\]
these data are used to learn the exposure-response functions, with the noise $\xi_{i\ell}$ specified in Section~\ref{subsec:estimation-fixed}, whereas the expectations $q_i(S)$ and $k_i^\pm(S)$ are computed by simulating the diffusion model.

Exposure counts provide a compact description of the final treatment state. The following assumption specifies how these counts enter an individual's outcome.

\begin{assumption}[Exposure-separable potential outcomes]
\label{assump:exposure-separable}
For each $i\in V$, there exist $\alpha_i\in\mathbb{R}$ and two functions
$f_i^+,f_i^-:\mathbb{Z}_{\ge 0}\to\mathbb{R}$ such that for all $z\in\{0,1\}^n$,
\begin{equation}
\label{eq:exposure-separable}
Y_i(z)
=
\alpha_i z_i
+
f_i^+\!\bigl(K_i^+(z)\bigr)
-
f_i^-\!\bigl(K_i^-(z)\bigr).
\end{equation}
Moreover, $f_i^\pm$ are nondecreasing and discretely concave:
for all integers $t\ge 1$,
\begin{equation}
\label{eq:discrete-concavity}
0 \le f_i^\pm(t+1)-f_i^\pm(t) \le f_i^\pm(t)-f_i^\pm(t-1).
\end{equation}
\end{assumption}

\paragraph{Example: information benefit and attention cost.}
Return to the trusted-source intervention in the Introduction. Let $z_i$ indicate whether user $i$ is active in the final state, and measure $Y_i(z)$ as their net benefit relative to no adoption. For a given user, let $N_i^+$ contain contacts whose corrections are relevant to their information needs, and let $N_i^-$ contain contacts whose notifications tend to be repetitive or irrelevant. These source sets are specified before the intervention. The term $\alpha_i z_i$ represents the benefit of the user's own adoption, $f_i^+(K_i^+(z))$ represents the benefit of useful information from active contacts, and $f_i^-(K_i^-(z))$ represents the associated attention cost. Information benefit can saturate as further contacts repeat facts the user has already learned. Attention cost can also saturate as the user habituates to repetitive notifications. One illustrative choice is
\[
f_i^\pm(k)
=b_i^\pm\bigl(1-\exp(-\lambda_i^\pm k)\bigr),
\qquad b_i^\pm\ge 0,\quad \lambda_i^\pm>0.
\]
Both functions are nondecreasing and discretely concave. The additive model treats own adoption, useful information, and attention cost as separate contributions to net benefit; within each source set, exposure counts summarize the effects of active contacts. Interactions between sources, threshold responses, and effects of exposure order require a richer outcome model. The reduction below relies on this additive separation; an interaction such as $z_iK_i^\pm(z)$ generally requires joint moments beyond $q_i(S)$ and $k_i^\pm(S)$.

Throughout, when $k_i^\pm(S)$ is non-integer, we interpret $f_i^\pm(k_i^\pm(S))$ via piecewise-linear interpolation
between adjacent integers; this preserves monotonicity and concavity and makes all expressions well-defined.
We use the expected activations and exposures to define the surrogate
\begin{equation}
\label{eq:surrogate-objective}
\widetilde F(S)
:=
\sum_{i\in V}\Bigl(
\alpha_i\,q_i(S)
+
f_i^+\!\bigl(k_i^+(S)\bigr)
-
f_i^-\!\bigl(k_i^-(S)\bigr)
\Bigr).
\end{equation}

The term $q_i(S)$ is essential: it retains the expected final activation of non-seed nodes.
Consequently, when $Y_i(z)=z_i$ and $f_i^\pm\equiv0$, $\widetilde F(S)=\sum_i q_i(S)$ recovers classical expected spread.

For estimation, nodes may be grouped into strata that share response functions.
Let $r(i)\in\{1,\dots,R_{\mathrm{str}}\}$ denote the stratum of node $i$, where $R_{\mathrm{str}}$ is the number of strata (distinct from the Monte Carlo sample count $R$ below); the fully heterogeneous case is recovered by taking one stratum per node.
For each stratum, we estimate $(\alpha_r,f_r^+,f_r^-)$ by shape-constrained weighted least squares. The estimator enforces monotonicity, discrete concavity, and zero response at zero exposure. Section~\ref{subsec:estimation-fixed} gives the full convex program.
For estimation, we additionally use the zero-baseline subclass $f_i^\pm(0)=0$, so $Y_i(\mathbf 0)=0$. This normalization only fixes the outcome baseline; the structural reduction does not require it.

\section{Framework and theoretical analysis}
\label{sec:method}

Evaluating the welfare $F(S)$ in \eqref{eq:of} requires both a causal data model and a tractable representation of diffusion-induced exposures.
Our analysis therefore proceeds in three steps:

(i) \textit{Identification and structural approximation.} We state conditions under which $F(S)$ is identified from replicated diffusion experiments and separately bound its approximation by an expected-exposure surrogate.\\
(ii) \textit{Estimation.} For any candidate seed set $S$, we construct a welfare estimator and characterize finite-sample and asymptotic error, including uniform expectation and high-probability bounds for data-dependent selection.\\
(iii) \textit{Optimization.} Given an optimizer with an approximation guarantee, we transfer the uniform evaluation bounds to the welfare of the selected seeds.

\subsection{Identification}
\label{subsec:identification}

We formalize identification in a replicated-experiment regime, where the analyst observes multiple independent replications of the same diffusion-outcome system. Write $\pi_\ell(S\mid X_\ell)$ for the logging probability of seed set $S$.

\begin{assumption}[Replicated experiments, conditional randomization, and positivity]
\label{assump:replication}
We observe $N$ independent replications with a common population law on a fixed graph. Let $X_\ell$ be pre-intervention context, $Z_\ell$ the logged seed set, and $\mathbf T_\ell(S)$ the potential noisy outcome vector under seed set $S$, with $\mathbb E[\sum_i T_{i\ell}(S)]=F(S)$. Observed outcomes satisfy $Y_{i\ell}=T_{i\ell}(Z_\ell)$. Conditional on $X_\ell$, $Z_\ell$ is independent of the collection $\{(z_{\infty,\ell}(S),\mathbf T_\ell(S)):S\in\mathcal S_K\}$. For every $S$ in the logging-supported target family $\mathcal S_{\rm log}\subseteq\mathcal S_K$, $\pi_\ell(S\mid X_\ell)>0$ almost surely.
\end{assumption}
Under Assumption~\ref{assump:replication}, inverse-probability weighting \citep{horvitz1952generalization,hernan2020whatif} identifies $F(S)$ for every logging-supported seed set (Proposition~\ref{prop:seedset-ips-id}; Appendix~\ref{app:proof-ips} treats estimated propensities). This identification is mainly conceptual: in a large network, any particular seed set is logged with tiny probability, so the weights are highly variable. We therefore use the exposure structure of Assumption~\ref{assump:exposure-separable}, which reduces welfare to low-dimensional quantities. Learning its responses requires the conditional mean and coverage conditions in Section~\ref{subsec:estimation-fixed}, and Stage~II approximates the propagation quantities by simulation.

\paragraph{Structural reduction via exposure.}
\label{subsec:structural-id}We now use the exposure mapping introduced in Section~\ref{sec:formulation}.
Write $\kappa_i^\pm:=\sup_{k\ge1}(-\Delta^2 f_i^\pm(k))$, where $\Delta^2 f_i^\pm(k):=f_i^\pm(k+1)-2f_i^\pm(k)+f_i^\pm(k-1)$.

The reduction yields a first-moment representation of welfare with a controlled approximation error (when full cascade samples are available, fitted outcomes can instead be averaged directly). The next theorem bounds this error and identifies the propagation conditions under which it becomes small.

\paragraph{Weak-propagation regime.}
\begin{assumption}[Low-probability propagation]
\label{assump:lpp}
Every directed edge has activation probability $0\le p_{ji}\le\varepsilon\le1$. The second-order statements concern $\varepsilon\to0$ with the graph and seed budget fixed.
\end{assumption}

Assumption~\ref{assump:lpp} is used only for the $O(\varepsilon^2)$ rate; the objective and the first part of Theorem~\ref{thm:exposure-compression} do not require it. Together with the live-edge condition in Appendix~\ref{app:diffusion-model}, this assumption controls the probability that several non-seed exposure sources for the same individual become active. The rate constants depend on graph-specific path multiplicity, exposure neighborhoods, and the seed budget; small edge probabilities alone do not ensure a numerically tight approximation.

For a node $i$ and seed set $S$, let
$
    U_i^\pm(S) := \sum_{j \in N_i^\pm \setminus S} z_{j, \infty}(S)
$
be the random \emph{non-seed} exposure increments, so that $K_i^\pm(S) = |N_i^\pm \cap S| + U_i^\pm(S)$, and let $(x)_2:=x(x-1)$ denote the falling factorial of order $2$. The factorial moment $\mathbb{E}[(U_i^\pm(S))_2]$ measures how often two or more non-seed sources of $i$ become active.

\begin{theorem}[Exposure reduction: general bound and weak-propagation rate]
\label{thm:exposure-compression}
Suppose Assumptions~\ref{assump:mono} and~\ref{assump:exposure-separable} hold.
Then for any node $i$ and seed set $S$,
\(
    \Bigl| \mathbb{E}[Y_i(z_\infty(S))]  - \alpha_i q_i(S) - f_i^+(k_i^+(S)) + f_i^-(k_i^-(S)) \Bigr|
     \le \frac{\kappa_i^+}{2}\mathbb{E}[(U_i^+(S))_2] + \frac{\kappa_i^-}{2}\mathbb{E}[(U_i^-(S))_2].
\)
Moreover, if Assumption~\ref{assump:lpp} and the live-edge Assumption~\ref{assump:live-edge} in Appendix~\ref{app:diffusion-model} also hold,
for any fixed budget $|S|\le K$, there exist constants $C_i^\pm(G,K)$
(depending on $G$, the exposure neighborhoods, and $K$) such that
\(
\mathbb{E}\!\bigl[(U_i^\pm(S))_2\bigr] \le C_i^\pm(G,K)\,\varepsilon^2,
\)
and hence
\begin{equation}
\label{eq:exposure-second-order}
\begin{aligned}
\mathbb{E}\bigl[Y_i(z_\infty(S))\bigr]
&=\alpha_i\,q_i(S)+f_i^+\!\bigl(k_i^+(S)\bigr)\\
&\quad-f_i^-\!\bigl(k_i^-(S)\bigr)+O(\varepsilon^2).
\end{aligned}
\end{equation}
\end{theorem}

The first part of Theorem~\ref{thm:exposure-compression} controls the Jensen gap created by replacing realized exposure counts with expected exposure counts. The second-order rate follows from the IC/live-edge moment condition in Assumption~\ref{assump:live-edge}.
The approximation error is controlled by (i) response curvature, which describes how quickly marginal effects diminish,
and (ii) the probability that two or more non-seed sources of the same node become active, captured by $\mathbb{E}[(U_i^\pm(S))_2]$. In the trusted-source example, the \emph{first few} beneficial exposures can provide most of the information gain; further exposures add less due to concavity.
Under the live-edge condition, low-probability propagation makes it rare that \emph{two} non-seed trusted neighbors both become active and jointly affect the same user, so these coincidences contribute at second order.

Strong propagation or dense path redundancy can make this factorial moment large and weaken the second-order control. The theorem identifies a regime in which expected activation and exposure quantities accurately summarize the outcome-relevant final state.

Summing \eqref{eq:exposure-second-order} over $i$ yields an exposure-based surrogate objective $\widetilde F(S)$ as in \eqref{eq:surrogate-objective}.
The next corollary formalizes the welfare-level approximation error uniformly over feasible seed sets.

\begin{corollary}[Uniform structural approximation]
\label{cor:welfare-approx}
Under Assumptions~\ref{assump:mono}, \ref{assump:exposure-separable}, \ref{assump:lpp}, and \ref{assump:live-edge},
\begin{equation}
\label{eq:welfare-approx}
\begin{aligned}
\sup_{S\in\mathcal S_K}|F(S)-\widetilde F(S)|
&\le B_{\rm str}\varepsilon^2,\\
B_{\rm str}
&:=\frac12\sum_i(\kappa_i^+C_i^++\kappa_i^-C_i^-).
\end{aligned}
\end{equation}
Here $C_i^\pm=C_i^\pm(G,K)$ are the constants in Theorem~\ref{thm:exposure-compression}.
\end{corollary}

The reduction is exact for linear responses or when each non-seed exposure increment takes values only in $\{0,1\}$ (Corollary~\ref{thm:point-id}); Corollary~\ref{thm:partial-id} states the resulting approximation interval for $F(S)$. When the two source sets of a node coincide ($N_i^+=N_i^-$, as in our synthetic outcome model; Appendix~\ref{app:synth-data}), we say that its exposure channels coincide. Only the net response $h_i=f_i^+-f_i^-$ is then identifiable under the design conditions of Section~\ref{subsec:estimation-fixed}, and the component curvatures in $B_{\rm str}$ give a model-dependent bound.

\subsection{Estimation and welfare evaluation}
\label{subsec:estimation-fixed}

Fix any candidate seed set $S\in\mathcal S_K$. We first estimate the exposure surrogate $\widetilde F(S)$ and then relate the estimation error to $F(S)$ through the structural bound. Uniform control over the candidate family will then permit data-dependent seed selection.

\textbf{Stage I: estimating response functions.}
For each stratum $r$ of $m_r$ nodes, we estimate the shared responses by the shape-constrained program \citep{groeneboom2014nonparametric,guntuboyina2018nonparametric}
\begin{equation}
\label{eq:shape-ls}
\begin{aligned}
\min_{\alpha_r,f_r^\pm}\quad
&\frac{1}{N_r}\sum_{\ell=1}^{N}\sum_{i:r(i)=r}w_{i\ell}\\
&\quad\times\bigl(Y_{i\ell}-\alpha_rz_{i\ell}\\
&\qquad-f_r^+(K_{i\ell}^+)+f_r^-(K_{i\ell}^-)\bigr)^2\\
&\quad+\lambda J(f_r^+,f_r^-).
\end{aligned}
\end{equation}
Here $N_r=Nm_r$ and $J=\operatorname{TV}(f_r^+)+\operatorname{TV}(f_r^-)\ge0$. For each sign and every interior grid point, impose
\begin{equation}
\label{eq:shape-constraints}
\begin{aligned}
0\le f_r^\pm(k+1)-f_r^\pm(k)
&\le f_r^\pm(k)-f_r^\pm(k-1),\\
f_r^\pm(0)&=0.
\end{aligned}
\end{equation}
For a grid with at least one positive point, also impose $f_r^\pm(1)\ge f_r^\pm(0)$. The finite grids cover every exposure attainable by the evaluated nodes. Let $\theta_r$ collect $\alpha_r$ and the response-grid values; $d_r$ is its dimension and $\mathcal C_r$ its feasible set. For coincident channels, $\theta_r$ uses the identifiable net response $h_r=f_r^+-f_r^-$; with $\lambda=0$, the component shape constraints then impose no restriction on the finite-grid net response. Appendix~\ref{app:shape-estimator} gives the feature construction and identifiable coordinates.

\paragraph{Statistical conditions.}
We condition on the complete observed design $\mathcal D$, including the nonnegative, design-measurable weights $w_{i\ell}\le w_{\max,r}$. Stacking the stratum-$r$ outcomes of replication $\ell$ into $\mathbf Y_{r\ell}$, we assume $\mathbf Y_{r\ell}=\Phi_{r\ell}\theta_r^\star+\xi_{r\ell}$ with feature matrix $\Phi_{r\ell}$ (Appendix~\ref{app:shape-estimator}) and $\theta_r^\star$ in the constraint set. Across replications, the vectors $\xi_{r\ell}$ are conditionally independent and mean zero; within a replication they may be dependent \citep{cameron2015cluster}. For the expectation bound, assume $\operatorname{Cov}(\xi_{r\ell}\mid\mathcal D)\preceq\sigma_r^2\gamma_r I$. The high-probability bound additionally assumes
\begin{equation}
\label{eq:cluster-mgf}
\mathbb E[e^{u^\top\xi_{r\ell}}\mid\mathcal D]
\le e^{\sigma_r^2\gamma_r\|u\|_2^2/2}\quad\text{for every }u.
\end{equation}
The factor $\gamma_r\ge1$ accounts for within-replication dependence. In the chosen identifiable coordinates, require $Q_r:=\sum_\ell\Phi_{r\ell}^\top W_{r\ell}\Phi_{r\ell}\succ0$ with $W_{r\ell}=\operatorname{diag}(w_{i\ell})$; a design satisfying this for every stratum is \emph{covered}. Define
\[
N_{{\rm eff},r}:=\frac{\lambda_{\min}(Q_r)}{\gamma_r w_{\max,r}}.
\]
Appendices~\ref{app:shape-estimator} and~\ref{app:design-conditions} give the features, the coincident-channel case, and the coverage requirement.

\begin{theorem}[Response estimation with clustered observations]
\label{thm:estimation}
Under the statistical conditions above, let $\widehat\theta_r$ solve \eqref{eq:shape-ls} with $\lambda=0$. Conditional on $\mathcal D$,
\[
\mathbb E[\|\widehat\theta_r-\theta_r^\star\|_2^2\mid\mathcal D]
\le\frac{4\sigma_r^2d_r}{N_{{\rm eff},r}}.
\]
Under \eqref{eq:cluster-mgf}, with conditional probability at least $1-e^{-t}$,
\begin{equation}
\label{eq:response-tail}
\begin{aligned}
\|\widehat\theta_r-\theta_r^\star\|_2&\le a_r(t),\\
a_r(t)&:=2\sigma_r\sqrt{\frac{d_r+2\sqrt{d_rt}+2t}{N_{{\rm eff},r}}}.
\end{aligned}
\end{equation}
\end{theorem}

Theorem~\ref{thm:estimation} shows that the response parameters are estimated at rate $\sigma_r\sqrt{d_r/N_{{\rm eff},r}}$: the error grows with the grid dimension and shrinks with the effective sample size. Its proof uses the constrained least-squares basic inequality and a sub-Gaussian vector bound \citep{hsu2012tail}.

\textbf{Stage II: Monte Carlo approximation of propagation.}
We assume that the diffusion model is correctly specified and can be simulated, for example by a live-edge sampler, and estimate the expected activations and exposures under each seed set by Monte Carlo \citep{owen2013montecarlo}. We draw a single bank of $R$ independent realizations, independent of the response noise, and reuse it for every candidate set. For fixed $S$, define
\begin{equation}
\label{eq:khat-mc}
\begin{aligned}
\widehat q_i(S)&=R^{-1}\sum_{b=1}^R z_{i,\infty}^{(b)}(S),\\
\widehat k_i^\pm(S)&=R^{-1}\sum_{b=1}^R K_i^{\pm,(b)}(S).
\end{aligned}
\end{equation}
Here $K_i^{\pm,(b)}(S)=K_i^\pm(z_\infty^{(b)}(S))$ and $z_\infty^{(b)}(S)$ is the final state in simulation $b$.
Then $\mathbb E[\widehat k_i^\pm(S)]=k_i^\pm(S)$ and $\mathbb E[\widehat q_i(S)]=q_i(S)$.

Under the IC/live-edge moment condition in Assumption~\ref{assump:live-edge} (Appendix~\ref{app:diffusion-model}), weak propagation also controls the variability of the exposure simulation: the non-seed exposure increments have first moment $O(\varepsilon)$, second factorial moment $O(\varepsilon^2)$, and hence $\mathrm{Var}(K_i^\pm(S))=O(\varepsilon)$ uniformly over feasible $S$ (Lemma~\ref{lem:moment-control}).
A Bernstein bound for \eqref{eq:khat-mc} gives leading simulation error $O(\sqrt{\varepsilon/R})$, with logarithmic factors and a finite-sample $1/R$ term (Proposition~\ref{prop:mc-concentration}); the full statements are in Appendix~\ref{app:uniform-mc}.

Define the plug-in estimator for the exposure surrogate:
\begin{equation}
\label{eq:Fhat-fixedS}
\widehat F(S)
:=
\sum_{i\in V}\Bigl(
\widehat\alpha_i\,\widehat q_i(S)
+
\widehat f_i^+\!\bigl(\widehat k_i^+(S)\bigr)
-
\widehat f_i^-\!\bigl(\widehat k_i^-(S)\bigr)
\Bigr),
\end{equation}
where $\widehat\alpha_i:=\widehat\alpha_{r(i)}$ and $\widehat f_i^\pm:=\widehat f_{r(i)}^\pm$ use the Stage I fit for node $i$'s stratum, and $\widehat q_i(S)$ and $\widehat k_i^\pm(S)$ come from Stage II. For coincident channels, the response terms combine into $\widehat h_i(\widehat k_i(S))$.

Evaluating $\widehat F$ thus combines estimated responses with simulated exposures. The decomposition below separates the resulting errors from the structural error, and Theorem~\ref{thm:fixedS-finite} bounds each term uniformly over $\mathcal S_K$.

For $L_i^\pm$-Lipschitz responses, adding and subtracting the true responses at simulated exposures gives
\begin{align}
|\widehat F(S)-F(S)|
&\le |F(S)-\widetilde F(S)| \nonumber\\
&\quad+\sum_i|\widehat\alpha_i-\alpha_i| \nonumber\\
&\quad+\sum_i\Bigl(
\|\widehat f_i^+-f_i^+\|_\infty
+\|\widehat f_i^--f_i^-\|_\infty
\Bigr) \nonumber\\
&\quad+\sum_i|\alpha_i|\,|\widehat q_i(S)-q_i(S)| \nonumber\\
&\quad+\sum_i\Bigl[
L_i^+|\widehat k_i^+(S)-k_i^+(S)| \nonumber\\
&\hspace{5.5em}
+L_i^-|\widehat k_i^-(S)-k_i^-(S)|
\Bigr].
\label{eq:welfare-component-errors}
\end{align}
The five terms are one structural, two response-estimation, and two simulation errors; the response norms are over the interpolation grids. For coincident channels, use one net-response norm and its corresponding Lipschitz-weighted exposure error. Taking conditional expectations gives the componentwise expectation bound.

Let $M=|\mathcal S_K|$. Applying Theorem~\ref{thm:estimation} and the feature-norm bound in Appendix~\ref{app:shape-estimator} gives the statistical error terms
\[
\begin{aligned}
\eta_{\rm stat}(\delta)
&=\sqrt3\sum_r m_r a_r\!\left(\log\frac{2R_{\rm str}}{\delta}\right),\\
\bar\eta_{\rm stat}
&=2\sqrt3\sum_r m_r\sigma_r\sqrt{\frac{d_r}{N_{{\rm eff},r}}}.
\end{aligned}
\]
The explicit Bernstein sums $\eta_{\rm MC}(\delta)$ and $\bar\eta_{\rm MC}$ in Appendix~\ref{app:uniform-mc} satisfy, for fixed graph and response bounds,
\[
\begin{aligned}
\eta_{\rm MC}(\delta)
&=O_{G,K}\!\left(
\sqrt{\frac{\varepsilon\log(12nM/\delta)}{R}}\right.\\
&\hspace{4.5em}\left.
+\frac{\log(12nM/\delta)}{R}\right).
\end{aligned}
\]
The expectation version replaces the logarithm by $\log(2M)+1$.

\begin{theorem}[Uniform welfare evaluation]
\label{thm:fixedS-finite}
Suppose the conditions of Corollary~\ref{cor:welfare-approx} and Theorem~\ref{thm:estimation} and the simulation conditions of Appendix~\ref{app:uniform-mc} hold, and the true responses are Lipschitz on the evaluation grid. Then, conditional on a covered design,
\begin{equation}
\label{eq:fixedS-bound}
\mathbb E\!\left[\sup_S|\widehat F(S)-F(S)|\mid\mathcal D\right]
\le B_{\rm str}\varepsilon^2+\bar\eta_{\rm stat}+\bar\eta_{\rm MC}.
\end{equation}
Under \eqref{eq:cluster-mgf}, with conditional probability at least $1-\delta$,
\begin{equation}
\label{eq:fixedS-rate}
\sup_S|\widehat F(S)-F(S)|
\le B_{\rm str}\varepsilon^2+\eta_{\rm stat}(\delta)+\eta_{\rm MC}(\delta).
\end{equation}
Both suprema are over $\mathcal S_K$. The expectation is over response noise and simulation.
\end{theorem}
Theorem~\ref{thm:fixedS-finite} decomposes the welfare error, uniformly over $S$ and conditional on a covered design, into a second-order structural error and two reducible errors from response estimation and simulation.
Under weak live-edge propagation, the structural error scales as $\varepsilon^2$ and the leading simulation fluctuation as $\sqrt{\varepsilon/R}$; the latter reflects the intrinsic variability of diffusion-induced exposures.

Along covered designs, the same decomposition gives consistency for $\widetilde F$ when $\bar\eta_{\rm stat}+\bar\eta_{\rm MC}\to0$, and for the true welfare when also $B_{\rm str}\varepsilon^2\to0$ (Corollary~\ref{thm:fixedS-asymp}); an unconditional guarantee must also account for coverage failure (Appendix~\ref{app:design-conditions}).

\paragraph{Seed selection.}
\label{subsec:optimization-theory}

We select seeds by greedily maximizing $\widehat F$: starting from the empty set, we estimate the marginal welfare gain of each candidate node and add the candidate with the largest gain until the budget is exhausted (Algorithm~\ref{alg:cim} in the appendix). For a normalized, monotone, submodular objective, greedy achieves the classical $(1-1/e)$ approximation guarantee \citep{nemhauser1978analysis}. Signed exposure responses can produce a nonmonotone objective, and we use greedy as an empirical heuristic in this general case. If the optimizer is $\rho$-approximate for $\widehat F$, Corollary~\ref{thm:end2end} converts this guarantee into one for the true welfare $F$ under uniform evaluation error; Corollary~\ref{cor:opt-rate} gives the expectation and high-probability versions.

\section{Experiments}
\label{sec:experiments}

\begin{figure}[t]
\centering
\includegraphics[width=\columnwidth]{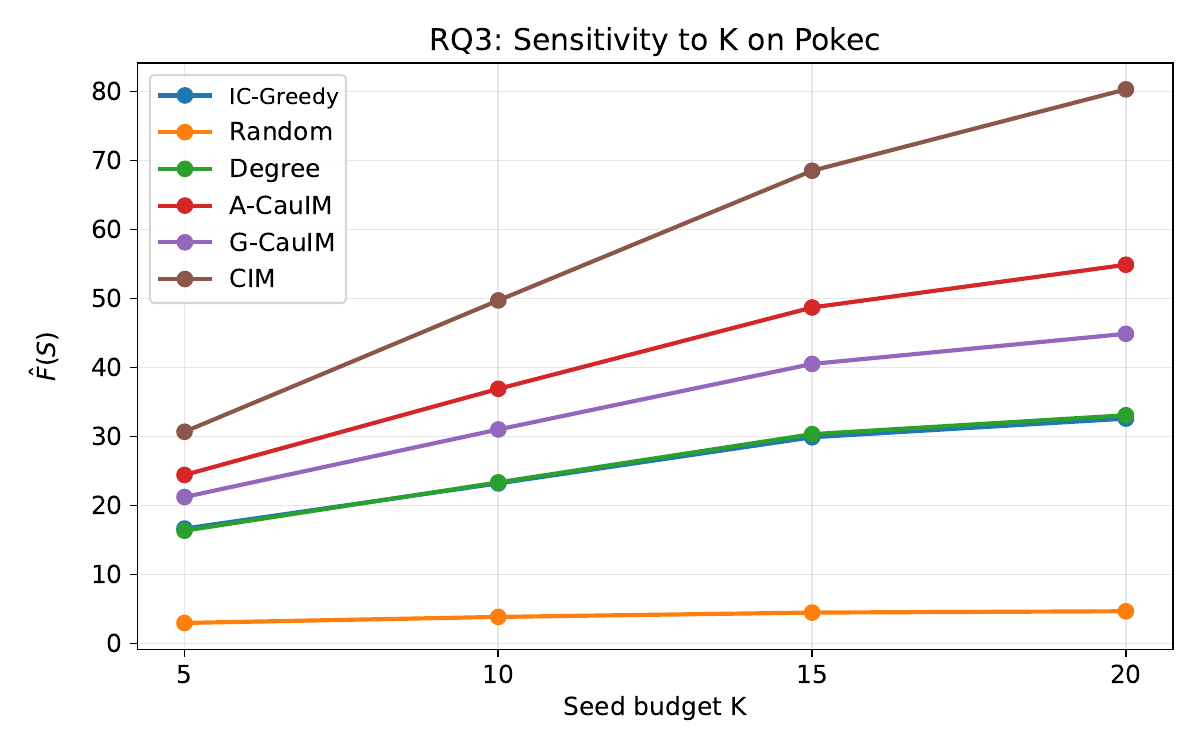}
\caption{RQ3: Budget sensitivity on Pokec. CIM has the highest welfare at each tested budget, and its margin grows with $K$. Complete results for all five benchmarks are in Figure~\ref{fig:rq3}.}
\label{fig:rq3-main}
\end{figure}

We evaluate CIM on GoodReads \citep{wan2018item,wan2019spoiler}, BlogCatalog \citep{tang2009relational}, Contact-Pri \citep{stehle2011contact,benson2018simplicial}, Pokec \citep{takac2012data}, and SD-100 \citep{su2023cauim}. The first four benchmarks provide network topology; diffusion-outcome observations are generated using the procedure in Appendix~\ref{app:synth-data}, while SD-100 uses a separate simulator. Methods are compared by the simulated steady-state welfare of their seed sets. We ask whether CIM improves steady-state welfare over the baselines (RQ1), how its performance changes with outcome noise and propagation strength (RQ2), and how it responds to the seed budget (RQ3). At $K=15$, we compare CIM with IC-Greedy, which maximizes expected spread \citep{kempe2003maximizing}, the individual treatment effect (ITE)-weighted methods A-CauIM and G-CauIM \citep{su2023cauim}, uniform random seeding, and degree-based seeding. Appendix~\ref{app:implementation} gives the implementation details.

\begin{table}[ht]
\caption{RQ1: mean welfare over 10 evaluation runs ($K=15$); the best value in each column is in bold.}
\label{tab:perf}
\begin{center}
\begin{small}
\begin{sc}
\resizebox{\columnwidth}{!}{
\begin{tabular}{lccccc}
\toprule
Methods & GoodReads & BlogCatalog & Contact-Pri & Pokec & SD-100 \\
\midrule
IC-Greedy & 2735.66 & 2128.08 & 156.00 & 29.88 & 89.01 \\
Random & $2706.25$ & $1090.12$ & $94.80$ & $4.47$ & $80.12$ \\
Degree & 2725.64 & 2412.85 & 149.00 & 30.33 & 89.12 \\
A-CauIM & 2700.66 & 2546.16 & 160.68 & 48.69 & 89.77 \\
G-CauIM & 2702.77 & 2501.61 & 161.44 & 40.51 & 86.41 \\
\textbf{CIM} & \textbf{2774.02} & \textbf{2878.89} & \textbf{164.00} & \textbf{68.51} & \textbf{90.67} \\
\bottomrule
\end{tabular}
}
\end{sc}
\end{small}
\end{center}
\end{table}

\paragraph{Data generation.}
On the four topology benchmarks, edges fire with probability $p_{ji}\sim\mathrm{Uniform}(0.01,0.05)$, both source sets are the in-neighborhood, and outcomes follow Assumption~\ref{assump:exposure-separable} with three strata, $f_r^+(k)=\beta_r^+\log(1+k)$, $f_r^-(k)=\beta_r^-\sqrt{k}$, and Gaussian noise with $\sigma=0.05$ (Appendix~\ref{app:synth-data}). CIM learns from $N=200$ logged replications, each seeded with $K=15$ nodes drawn uniformly at random. As implementation heuristics, its greedy search uses lazy evaluation and only considers the 200 nodes with the highest degree-weighted exposure score (Appendix~\ref{app:implementation}).

\subsection{General Performance (RQ1)}
Table~\ref{tab:perf} shows that CIM obtains the highest welfare on all five datasets. The advantage is largest on sparse graphs with heterogeneous exposure: relative to Degree, CIM improves welfare by approximately $19.3\%$ on BlogCatalog and $125.9\%$ on Pokec. This is consistent with the motivation in Section~\ref{sec:intro}: exposure-aware welfare optimization is most useful when reach-based or ITE-weighted methods waste seeds on saturated neighborhoods.

\subsection{Robustness and Budget Sensitivity (RQ2--RQ3)}
RQ2 uses GoodReads with $K=20$ and $N\in\{5,10,20,50\}$ logged replications. We vary the outcome noise level up to $\sigma=0.4$ and, separately, scale all edge activation probabilities by up to $20\%$, which weakens the low-propagation approximation. The oracle gap $1-F(\widehat S)/F(S^\star_{\rm oracle})$ compares CIM with greedy selection under the ground-truth responses (Appendix~\ref{app:implementation}). The gaps generally increase with noise or propagation strength but stay below $0.15\%$ and $0.7\%$, respectively, and more logged replications generally reduce them (Figure~\ref{fig:rq2} in the appendix).

For RQ3, Figure~\ref{fig:rq3-main} shows budget sensitivity on Pokec for $K\in\{5,10,15,20\}$: CIM leads at each tested budget, and its margin over the baselines grows with $K$. In contrast, the near-saturated SD-100 benchmark (100 nodes, mean degree 24) shows smaller differences; Figure~\ref{fig:rq3} in the appendix gives all five benchmarks.

\section{Conclusion}
\label{sec:discussion}

We studied seed selection when activation propagates through a network and outcomes are measured after diffusion stabilizes. Classical expected spread is recovered when final activation is the outcome, while exposure responses allow heterogeneous value, saturation, and adverse spillovers. Our main result gives a deterministic curvature bound for replacing random final exposures with their expectations. Weak live-edge propagation turns this bound into a second-order approximation. The reduced objective can then be learned from logged diffusion-outcome observations and used for welfare-aware seed selection.

\subsection*{AI use statement}
In this work, we used generative AI tools only to polish the writing, for example to improve the grammar, wording, and readability of text written by the authors. We did not use generative AI tools to generate research ideas, develop the methods or theory, write proofs or code, run experiments, or analyze results. We have reviewed all AI-assisted edits and take responsibility for the final content of this work, including any text produced with the aid of generative AI.

\bibliography{references}
\bibliographystyle{apalike}


\clearpage

\clearpage
\appendix
\onecolumn
\section{Notation}
\label{app:notation}
The table distinguishes population quantities, observed exposures, and simulation estimates. Constants in graph-dependent rates may also depend on fixed response bounds.
\begingroup
\small
\begin{longtable}{@{}p{.26\textwidth}p{.69\textwidth}@{}}
\toprule
Symbol & Meaning \\
\midrule
\endfirsthead
\toprule
Symbol & Meaning (continued) \\
\midrule
\endhead
\bottomrule
\endfoot
$G=(V,E)$, $n$ & Fixed directed network and number of nodes. \\
$S$, $K$, $\mathcal S_K$ & Seed set, budget, and all seed sets with at most $K$ nodes. \\
$M$ & Cardinality of the evaluated family $\mathcal S_K$. \\
$z_{i,t}(S)$, $z_\infty(S)$ & Node activation at time $t$ and final activation vector under $S$. \\
$Y_i(z)$, $F(S)$ & Mean potential-outcome response and steady-state welfare. \\
$N_i^\pm$, $M_i^\pm$ & Prespecified beneficial/adverse source sets and their sizes. \\
$K_i^\pm(S)$ & Random realized final exposure counts. \\
$k_i^\pm(S)$, $q_i(S)$ & Expected exposure counts and final activation probability. \\
$U_i^\pm(S)$ & Non-seed exposure increment: $K_i^\pm(S)-|N_i^\pm\cap S|$. \\
$(u)_2$ & Second falling factorial $u(u-1)$. \\
$\alpha_i$, $f_i^\pm$ & Own-activation coefficient and monotone, discretely concave exposure responses. \\
$\bar f$ & Piecewise-linear interpolation of a response defined on an integer grid. \\
$\kappa_i^\pm$ & Supremum of the nonnegative discrete curvature of $f_i^\pm$. \\
$\varepsilon$, $p_{ji}$ & Weak-propagation bound and edge activation probability. \\
$\zeta$ & Relative propagation perturbation in the robustness experiment. \\
$L_e$, $G_L$ & Bernoulli live-edge indicator and sampled live-edge graph. \\
$D_i^\pm$, $C_i^\pm$ & Fixed-graph constants controlling first and second factorial exposure moments. \\
$\widetilde F(S)$ & Deterministic expected-exposure surrogate for $F(S)$. \\
$B_{\rm str}$ & Uniform structural-error coefficient in $|F-\widetilde F|\le B_{\rm str}\varepsilon^2$. \\
$N$, $\ell$ & Number and index of independent diffusion replications. \\
$X_\ell$, $Z_\ell$ & Pre-intervention context and logged seed assignment. \\
$\mathbf T_\ell(S)$ & Potential noisy outcome vector in replication $\ell$ under $S$. \\
$Y_{i\ell}$, $K_{i\ell}^\pm$ & Observed noisy response and observed final exposure counts. \\
$\pi_\ell$, $\widehat\pi_\ell$ & True and estimated logging propensities. \\
$\mathcal S_{\rm log}$ & Target seed sets with positive probability under the logging design. \\
$W_\ell(S)$ & Seed-set IPS weight $\mathbf1\{Z_\ell=S\}/\pi_\ell(S\mid X_\ell)$. \\
$r(i)$, $R_{\rm str}$, $m_r$ & Stratum assignment, number of strata, and nodes in stratum $r$. \\
$N_r$ & Number of pooled node observations in stratum $r$, equal to $Nm_r$ for complete observations. \\
$B_r^\pm$, $d_r$ & Exposure-grid bounds and dimension of the chosen identifiable parameterization. \\
$h_r$ & Net response $f_r^+-f_r^-$ for coincident exposure channels. \\
$\theta_r^\star$, $\widehat\theta_r$ & True and fitted identifiable response parameters, including own activation. \\
$\mathbf b_B$, $\phi_r$, $\Phi_{r\ell}$ & Interpolation basis, individual feature vector, and observed design matrix. \\
$\mathcal C_r$ & Convex feasible set for the response parameters. \\
$w_{i\ell}$, $W_{r\ell}$ & Node regression weights and their diagonal matrix. \\
$w_{\max,r}$ & Upper bound for regression weights in stratum $r$. \\
$\mathcal D$ & Complete observed design and its design-measurable weights. \\
$\mathcal O$ & Logged data supplied to the algorithm, including responses. \\
$\xi_{r\ell}$, $\sigma_r^2$, $\gamma_r$ & Replication-level noise vector, noise scale, and within-replication dependence factor. \\
$Q_r$ & Weighted information matrix $\sum_\ell\Phi_{r\ell}^\top W_{r\ell}\Phi_{r\ell}$. \\
$N_{{\rm eff},r}$ & Effective information $\lambda_{\min}(Q_r)/(\gamma_r w_{\max,r})$ under covered designs. \\
$\lambda$, $J$ & Regularization coefficient and total-variation penalty; main rates use $\lambda=0$. \\
$a_r(t)$ & High-probability parameter-error radius at failure level $e^{-t}$. \\
$R$, $b$ & Number and index of independent propagation simulations (distinct from $R_{\rm str}$). \\
$\widehat q_i$, $\widehat k_i^\pm$ & Monte Carlo estimates of activation probabilities and expected exposures. \\
$\widehat F(S)$ & Welfare estimator using fitted responses and simulated propagation quantities. \\
$\mathcal J_i$, $M_{ic}$, $L_{ic}$ & Evaluation channels, their exposure bounds, and true-response Lipschitz constants. \\
$v_{iq}$, $v_{ic}$ & Uniform activation and exposure variance bounds over feasible seed sets. \\
$b_R(v,b,t)$ & Bernstein radius $\sqrt{2vt/R}+2bt/(3R)$. Its second argument is a range bound. \\
$\eta_{\rm stat}$, $\eta_{\rm MC}$ & High-probability bounds for uniform response and simulation errors. \\
$\bar\eta_{\rm stat}$, $\bar\eta_{\rm MC}$ & Bounds for the expected uniform errors. \\
$\delta$, $\delta_{\rm des}$ & Conditional tail-failure level and design-coverage failure probability. \\
$S^\star$, $\widehat S$, $\rho$ & Welfare-maximizing seed set, selected seed set, and assumed optimization approximation factor. \\
\end{longtable}
\addtocounter{table}{-1}
\endgroup

\section{Limitations}
\label{app:limitations}
The structural rate is a fixed-graph, fixed-budget weak-propagation result. Path multiplicity enters its constants, so small edge probabilities can coexist with a loose bound in a highly connected graph. The response model assumes additive, monotone concave exposure channels. Threshold responses, interactions between specific sources, and effects of exposure order require additional modeling.

The statistical bounds are conditional on a correctly specified response model and a design covering the evaluated coordinates. Dependence within a replication can substantially reduce effective information. Coincident signed channels identify a net response; they do not separately identify its two components. The diffusion model is assumed to be correctly specified; misspecification would add a further evaluation error. Empirical performance alone does not establish the required coverage or noise conditions. The reported benchmarks use generated outcomes and should be interpreted within their simulation settings.

The seed-selection transfer requires an optimizer with an approximation guarantee. For a general signed objective, greedy selection and candidate pruning are empirical heuristics without such a guarantee. Extending the analysis to changing graphs, non-progressive diffusion, uncertain propagation models, and dependent experimental rounds remains future work.

\section{Literature review}
\label{app:literature}
The main text gives the closest comparisons. This appendix places the welfare objective, exposure assumptions, and statistical analysis within the broader literature.

\paragraph{Network interventions and outcomes.}
Online mobilization experiments demonstrate that an intervention can affect both its recipients and their social contacts \citep{bond2012mobilization}. Network targeting has also been evaluated in health-product diffusion \citep{kim2015targeting} and school interventions aimed at changing conflict norms \citep{paluck2016conflict}. These studies motivate the distinction between the initially selected participants, subsequent exposure, and the outcome used to evaluate an intervention.

\paragraph{Expected spread and scalable seed selection.}
The classical influence-maximization formulation optimizes expected final cascade size under a specified propagation model \citep{kempe2003maximizing}. The monotone submodularity of expected spread connects it to greedy maximization under a cardinality constraint \citep{nemhauser1978analysis}. Subsequent work develops efficient diffusion estimators and algorithms, including practical heuristics \citep{chen2010scalable}, reverse-reachable sampling \citep{borgs2014nearlyoptimal,tang2014tim}, and martingale-based refinements \citep{tang2015martingale}. The present formulation contains expected spread through $Y_i(z)=z_i$. With nonnegative heterogeneous coefficients and zero exposure terms, it also contains weighted spread. Nonlinear outcome responses require evaluating the welfare associated with the final activation pattern.

\paragraph{Interference and exposure mappings.}
Potential-outcome methods define direct and spillover effects under interference \citep{hudgens2008toward}. Exposure mappings specify which aspects of an assignment are relevant to a target estimand \citep{aronow2017estimating}; graph cluster randomization uses network structure to control exposure probabilities and estimator variance \citep{ugander2013graph}. Related network experiments study design and bias reduction \citep{eckles2017design}, exact randomization tests \citep{athey2018exact}, and causal peer influence \citep{toulis2013peer}. Arbitrary interference can make estimation intrinsically difficult even in large experiments \citep{basse2018limitations}. Our exposure-separable assumption is a specific restriction on the final-state response. The curvature bound quantifies a further approximation: replacing its random exposure argument by an expectation.

\paragraph{Approximate interference and temporal decisions.}
Approximate neighborhood interference allows the effect of distant assignments to diminish while remaining nonzero \citep{leung2022approximate}. Its network-size asymptotics address a different approximation from our fixed-graph weak-propagation analysis. Sequential experiments under interference also explicitly combine exposure mappings with online decision making \citep{zhang2025online}. That work studies the estimation-regret trade-off across experimental rounds. We choose an initial seed set and evaluate outcomes after the induced diffusion reaches its final state. The data used for learning are independent replications of that system.

\paragraph{Welfare and marketplace interference.}
Heterogeneous causal gains motivate objectives beyond cascade size \citep{su2023cauim}, and policy targeting under interference connects treatment allocation to welfare \citep{viviano2025policy}. Two-sided marketplace models show how interference affects experimental bias and how the choice of randomization depends on market balance \citep{johari2022experimental}. An Airbnb meta-experiment provides evidence that cluster randomization can reduce interference bias \citep{holtz2024reducing}. Such applications motivate careful target and assignment definitions. Applying our bounds additionally requires the specified progressive diffusion, response, and observation conditions.

\paragraph{Statistical tools and optimization scope.}
Shape-restricted regression supplies convex estimators for constrained response functions \citep{groeneboom2014nonparametric,guntuboyina2018nonparametric}. Inverse-probability methods require both assignment support and causal identification assumptions \citep{horvitz1952generalization,hernan2020whatif}; augmented estimating equations are a separate statistical construction \citep{robins1994estimation}. Network-correlated outcomes and clustered errors require their own dependence analysis \citep{basse2018correlated,cameron2015cluster}. Our response bound states conditional mean, covariance, tail, and design conditions explicitly and uses the vector inequality of \citet{hsu2012tail}. Monte Carlo evaluation uses independent realizations and may share random numbers across candidate seed sets \citep{owen2013montecarlo}. Classical lazy evaluation applies when valid marginal-gain upper bounds are available \citep{minoux1978accelerated}. The unconstrained nonnegative submodular problem admits a different approximation analysis \citep{buchbinder2015tight}; for the general signed, budget-constrained welfare problem, the transfer result in Corollary~\ref{thm:end2end} therefore assumes an approximation guarantee for the optimizer.

\section{Proofs and statistical conditions}
\label{app:proofs}
\subsection{Auxiliary theoretical statements}
\label{app:aux-statements}

\begin{proposition}[Identification with true assignment probabilities]
\label{prop:seedset-ips-id}
Under Assumption~\ref{assump:replication}, for $S\in\mathcal S_{\rm log}$ and integrable potential outcomes, define $W_\ell(S)=\mathbf1\{Z_\ell=S\}/\pi_\ell(S\mid X_\ell)$. Then
\[
F(S)=\mathbb E[W_\ell(S)\sum_iY_{i\ell}].
\]
The empirical average $N^{-1}\sum_\ell W_\ell(S)\sum_iY_{i\ell}$ using the true propensities is unbiased. Estimated propensities are covered by the separate consistency argument in Appendix~\ref{app:proof-ips}.
\end{proposition}

\begin{corollary}[Transfer from estimated welfare to seed selection]
\label{thm:end2end}
Suppose an algorithm returns $\widehat S$ satisfying
\begin{equation}
\label{eq:alg-rho}
\widehat F(\widehat S)\ge\rho\max_{S\in\mathcal S_K}\widehat F(S),
\qquad 0<\rho\le1.
\end{equation}
If $\sup_{S\in\mathcal S_K}|\widehat F(S)-F(S)|\le\eta$, then
\begin{equation}
\label{eq:end2end}
F(\widehat S)\ge\rho\max_{S\in\mathcal S_K}F(S)-(1+\rho)\eta.
\end{equation}
This is a deterministic implication; condition~\eqref{eq:alg-rho} is an assumption on the optimizer, not a result.
\end{corollary}

Under the zero-baseline IC specification, the feasible empty seed set satisfies $\widehat F(\varnothing)=0$, so $\max_{S\in\mathcal S_K}\widehat F(S)\ge0$.

\begin{corollary}[Structural approximation interval]
\label{thm:partial-id}
Under Corollary~\ref{cor:welfare-approx},
\[
\begin{aligned}
F(S)&\in[\widetilde F(S)-B_{\rm str}\varepsilon^2,\,
\widetilde F(S)+B_{\rm str}\varepsilon^2]\\
&\quad\text{for every }S\in\mathcal S_K.
\end{aligned}
\]
The interval is a model-based enclosure. An enclosure determined by observed data additionally requires identification of the surrogate and of a valid radius. No sharp identified set is claimed.
\end{corollary}

\begin{corollary}[Exact reduction in special cases]
\label{thm:point-id}
Under the general part of Theorem~\ref{thm:exposure-compression}, if either (i) both response curves have zero discrete curvature, or (ii) $(U_i^\pm(S))_2=0$ almost surely for every node, sign, and feasible seed set, then $F(S)=\widetilde F(S)$ for every $S$. Identification from observations still requires data assumptions for the response and propagation quantities.
\end{corollary}

\begin{lemma}[Moments of non-seed exposure]
\label{lem:moment-control}
Under Assumptions~\ref{assump:lpp} and \ref{assump:live-edge}, for fixed finite $G,K$, there exist finite constants $D_i^\pm(G,K)$ and $C_i^\pm(G,K)$ such that, uniformly over $\mathcal S_K$,
\[
\mathbb E[U_i^\pm(S)]\le D_i^\pm\varepsilon,\qquad
\mathbb E[(U_i^\pm(S))_2]\le C_i^\pm\varepsilon^2.
\]
Hence $\operatorname{Var}(K_i^\pm(S))\le(D_i^\pm+C_i^\pm)\varepsilon$ for $0\le\varepsilon\le1$.
\end{lemma}

\begin{proposition}[Monte Carlo concentration for a fixed seed set]
\label{prop:mc-concentration}
For $R$ independent simulations at $S$, with $0\le K_i^\pm\le M_i^\pm:=|N_i^\pm|$, for every $t>0$,
\begin{equation}
\label{eq:mc-bernstein}
\mathbb P\!\left(
\begin{aligned}
|\widehat k_i^\pm(S)-k_i^\pm(S)|
&>
\sqrt{\frac{2\operatorname{Var}(K_i^\pm(S))t}{R}}\\
&\quad+\frac{2M_i^\pm t}{3R}
\end{aligned}
\right)\le2e^{-t}.
\end{equation}
The same bound applies to $\widehat q_i$ with range one. The finite-sample linear term is retained.
\end{proposition}

\begin{corollary}[Consistency along covered designs]
\label{thm:fixedS-asymp}
Consider a sequence of designs satisfying the statistical conditions. If $\bar\eta_{\rm stat}+\bar\eta_{\rm MC}\to0$, then
\[
\sup_S|\widehat F(S)-\widetilde F(S)|\xrightarrow{p}0
\]
under the corresponding conditional noise and simulation laws. If also $B_{\rm str}\varepsilon^2\to0$, then $\sup_S|\widehat F(S)-F(S)|\xrightarrow{p}0$. Here the target may vary along the sequence. Unconditional consistency follows if design coverage holds with probability tending to one and the same errors vanish in probability.
\end{corollary}

\begin{corollary}[Seed-selection guarantees in expectation and probability]
\label{cor:opt-rate}
Suppose the approximation property \eqref{eq:alg-rho} holds for every realization and the conditions of Theorem~\ref{thm:fixedS-finite} hold. Conditional on a covered design,
\begin{align}
\mathbb E[F(\widehat S)\mid\mathcal D]
&\ge \rho\max_S F(S) \nonumber\\
&\quad-(1+\rho)\bigl(B_{\rm str}\varepsilon^2+\bar\eta_{\rm stat}\bigr)
\nonumber\\
&\quad-(1+\rho)\bar\eta_{\rm MC}.
\label{eq:policy-expectation}
\end{align}
Under the joint sub-Gaussian condition, with conditional probability at least $1-\delta$,
\begin{equation}
\label{eq:policy-highprob}
\begin{aligned}
F(\widehat S)&\ge \rho\max_S F(S)\\
&\quad-(1+\rho)\bigl(B_{\rm str}\varepsilon^2+\eta_{\rm stat}(\delta)\bigr)\\
&\quad-(1+\rho)\eta_{\rm MC}(\delta).
\end{aligned}
\end{equation}
The maxima are over $\mathcal S_K$. If a restricted search supplies the approximation property only within a smaller family, the corresponding comparison is within that family.
\end{corollary}

\subsection{Preliminaries: discrete concavity, interpolation, and curvature bounds}
\label{app:prelim}

Throughout, let $f:\mathbb{Z}_{\ge 0}\to\mathbb{R}$ be a nondecreasing discretely concave function, i.e.,
\[
0\le \Delta f(t):=f(t+1)-f(t)\le \Delta f(t-1)\qquad \forall t\ge 1.
\]
Equivalently, the discrete second difference is nonpositive:
\[
\Delta^2 f(t):=f(t+1)-2f(t)+f(t-1)\le 0,\qquad t\ge 1.
\]
Recall the curvature parameter
\[
\begin{aligned}
\kappa(f)
&:=\sup_{t\ge 1}\bigl(\Delta f(t-1)-\Delta f(t)\bigr)\\
&=\sup_{t\ge 1}\bigl(-\Delta^2 f(t)\bigr)\in[0,\infty).
\end{aligned}
\]

\paragraph{Piecewise-linear extension.}
Define the piecewise-linear interpolation $\bar f:[0,\infty)\to\mathbb{R}$ by
\begin{equation}
\label{eq:interp}
\begin{aligned}
\bar f(x)&:=(1-\theta) f(m)+\theta f(m+1),\\
&\quad x=m+\theta,\quad m\in\mathbb{Z}_{\ge 0},\quad\theta\in[0,1).
\end{aligned}
\end{equation}
In the main text we write $f(x)$ for $\bar f(x)$ whenever the argument is not an integer.

\begin{lemma}[Concavity of the interpolation]
\label{lem:interp-concave}
If $f$ is discretely concave and nondecreasing on $\mathbb{Z}_{\ge 0}$, then $\bar f$ defined in \eqref{eq:interp}
is concave and nondecreasing on $[0,\infty)$.
Moreover, on each interval $[m,m+1]$, $\bar f$ has slope $\Delta f(m)$, and these slopes are nonincreasing in $m$.
\end{lemma}

\begin{proof}
By definition, $\bar f$ is linear on each unit interval $[m,m+1]$ with slope $\Delta f(m)$.
Discrete concavity implies $\Delta f(m)\ge \Delta f(m+1)$ for all $m\ge 0$, hence the slopes of the linear pieces are nonincreasing.
A piecewise-linear function with nonincreasing slopes is concave.
Nondecreasingness follows from $\Delta f(m)\ge 0$ for all $m$.
\end{proof}

The next lemma is a discrete analogue of a second-order Taylor bound: concavity implies an ``upper tangent'' linear bound,
and curvature controls the departure from linearity through a quadratic term involving the falling factorial $(u)_2=u(u-1)$.

\begin{lemma}[Discrete second-order bound controlled by curvature]
\label{lem:discrete-taylor}
Let $f:\mathbb{Z}_{\ge 0}\to\mathbb{R}$ be discretely concave with curvature $\kappa=\kappa(f)$.
Then for any integers $t\ge 0$ and $u\ge 0$,
\begin{equation}
\label{eq:upper-linear}
f(t+u)\le f(t)+u\,\Delta f(t),
\end{equation}
and
\begin{equation}
\label{eq:lower-quadratic}
f(t+u)\ge f(t)+u\,\Delta f(t)-\frac{\kappa}{2}\,(u)_2.
\end{equation}
\end{lemma}

\begin{proof}
For \eqref{eq:upper-linear}, note that discrete concavity implies $\Delta f(t+r)\le \Delta f(t)$ for all $r\ge 0$.
Hence
\[
f(t+u)-f(t)=\sum_{r=0}^{u-1}\Delta f(t+r)\le \sum_{r=0}^{u-1}\Delta f(t)=u\,\Delta f(t).
\]

For \eqref{eq:lower-quadratic}, define the slope drops
\[
\delta_s:=\Delta f(t+s-1)-\Delta f(t+s)\ge 0,\qquad s\ge 1.
\]
Then $\delta_s\le \kappa$ by definition of curvature.
Also, for $r\ge 0$,
\(
\Delta f(t+r)=\Delta f(t)-\sum_{s=1}^{r}\delta_s
\)
(with the empty sum equal to $0$).
Therefore
\begin{align*}
f(t+u)-f(t)
&=\sum_{r=0}^{u-1}\Delta f(t+r)\\
&\quad=\sum_{r=0}^{u-1}\Bigl(\Delta f(t)-\sum_{s=1}^{r}\delta_s\Bigr)\\
&=u\,\Delta f(t)-\sum_{r=0}^{u-1}\sum_{s=1}^{r}\delta_s\\
&\quad=u\,\Delta f(t)-\sum_{s=1}^{u-1}(u-s)\delta_s\\
&\ge u\,\Delta f(t)-\kappa\sum_{s=1}^{u-1}(u-s)\\
&\quad=u\,\Delta f(t)-\kappa\sum_{r=1}^{u-1}r\\
&=u\,\Delta f(t)-\frac{\kappa}{2}\,u(u-1)\\
&\quad=u\,\Delta f(t)-\frac{\kappa}{2}\,(u)_2,
\end{align*}
which proves \eqref{eq:lower-quadratic}.
\end{proof}

We now convert Lemma~\ref{lem:discrete-taylor} into a curvature-controlled Jensen gap bound.

\begin{lemma}[Curvature-controlled Jensen gap via factorial moment]
\label{lem:jensen-gap}
Let $U$ be a $\mathbb{Z}_{\ge 0}$-valued random variable with $\mu:=\mathbb{E}[U]$.
Let $f$ be discretely concave with curvature $\kappa=\kappa(f)$, and let $\bar f$ be its interpolation.
Then for any integer $t\ge 0$,
\begin{equation}
\label{eq:jensen-gap}
0\le \bar f(t+\mu)-\mathbb{E}\bigl[f(t+U)\bigr]\le \frac{\kappa}{2}\,\mathbb{E}\bigl[(U)_2\bigr].
\end{equation}
\end{lemma}

\begin{proof}
The nonnegativity $\bar f(t+\mu)\ge \mathbb{E}[f(t+U)]$ follows from concavity of $\bar f$ (Lemma~\ref{lem:interp-concave}) and Jensen's inequality:
\[
\bar f(t+\mu)=\bar f\bigl(\mathbb{E}[t+U]\bigr)\ge \mathbb{E}\bigl[\bar f(t+U)\bigr]=\mathbb{E}\bigl[f(t+U)\bigr],
\]
where the last equality holds because $t+U$ is integer-valued.

For the upper bound, by concavity of $\bar f$, the line of slope $\Delta f(t)$ from $t$ is a supporting line for all $x\ge t$
(because slopes of $\bar f$ on $[m,m+1]$ equal $\Delta f(m)$ and are nonincreasing in $m$).
Hence
\begin{equation}
\label{eq:supporting}
\begin{aligned}
\bar f(t+\mu)
&\le \bar f(t)+\mu\,\Delta f(t)\\
&=f(t)+\mu\,\Delta f(t).
\end{aligned}
\end{equation}
On the other hand, applying Lemma~\ref{lem:discrete-taylor} to the integer $U$ and taking expectations yields
\[
\begin{aligned}
\mathbb{E}\bigl[f(t+U)\bigr]
&\ge f(t)+\mu\,\Delta f(t)\\
&\quad-\frac{\kappa}{2}\,\mathbb{E}\bigl[(U)_2\bigr].
\end{aligned}
\]
Subtracting this inequality from \eqref{eq:supporting} proves \eqref{eq:jensen-gap}.
\end{proof}

\subsection{Proof of Proposition~\ref{prop:seedset-ips-id} and estimated propensities}
\label{app:proof-ips}
\begin{proof}
Fix $S\in\mathcal S_{\rm log}$ and set $T_\ell(S)=\sum_iT_{i\ell}(S)$. By consistency and positivity,
\begin{align*}
\mathbb E[W_\ell(S)\sum_iY_{i\ell}\mid X_\ell]
&=\mathbb E[T_\ell(S)\mid X_\ell,Z_\ell=S]\\
&=\mathbb E[T_\ell(S)\mid X_\ell].
\end{align*}
The second equality uses conditional randomization, as in the standard IPW argument \citep{hernan2020whatif}. Taking expectations gives $F(S)$; linearity of expectation gives exact unbiasedness for the empirical average with true propensities.
\end{proof}

\paragraph{Plug-in propensities.}
For a fixed supported target $S$, additionally suppose $\pi_\ell(S\mid X_\ell)\ge p_0>0$, the independent true-IPW summands have uniformly bounded second moments, and $\mathbb E|T_\ell(Z_\ell)|$ is uniformly bounded. Let
\[
e_N=\max_{\ell\le N}|\widehat\pi_\ell(S\mid X_\ell)-\pi_\ell(S\mid X_\ell)|\xrightarrow{p}0.
\]
On the event $e_N\le p_0/2$, the difference between the plug-in and true-IPW estimates is at most
\[
\frac{2e_N}{p_0^2}\,
\frac1N\sum_{\ell=1}^N|T_\ell(Z_\ell)|=o_p(1).
\]
The true-IPW average converges to $F(S)$ by its variance bound, and the complement of this event has probability tending to zero. This proves plug-in consistency under the stated conditions. Finite-sample unbiasedness is generally lost. Learning propensity weights from the response outcomes also requires checking the conditional mean and weight conditions of Theorem~\ref{thm:estimation}; the IPS identity alone supplies no such regression guarantee. Augmented estimating equations are a separate construction \citep{robins1994estimation}.

\subsection{Proof of Theorem~\ref{thm:exposure-compression} (structural reduction)}
\label{app:proof-structural}

\begin{proof}[Proof of Theorem~\ref{thm:exposure-compression}]
Fix $i$ and $S$.
Under Assumption~\ref{assump:exposure-separable}, for any activation state $z$,
\[
Y_i(z)=\alpha_i z_i+f_i^+\bigl(K_i^+(z)\bigr)-f_i^-\bigl(K_i^-(z)\bigr),
\]
Evaluating at $z=z_\infty(S)$ and taking expectation gives
\begin{equation}
\label{eq:Yi-exp}
\begin{aligned}
\mathbb{E}[Y_i(z_\infty(S))]
&=\alpha_i\,q_i(S)\\
&\quad+\mathbb{E}\Bigl[f_i^+\bigl(K_i^+(S)\bigr)\Bigr]\\
&\quad-\mathbb{E}\Bigl[f_i^-\bigl(K_i^-(S)\bigr)\Bigr].
\end{aligned}
\end{equation}

Write $K_i^\pm(S)=|N_i^\pm\cap S|+U_i^\pm(S)$ where
\(
U_i^\pm(S)=\sum_{j\in N_i^\pm\setminus S} z_{j,\infty}(S)
\)
.
Let $\mu_i^\pm:=\mathbb{E}[U_i^\pm(S)]$ and note that
\[
k_i^\pm(S)=\mathbb{E}[K_i^\pm(S)]=|N_i^\pm\cap S|+\mu_i^\pm.
\]
Apply Lemma~\ref{lem:jensen-gap} to $f_i^\pm$ with $t=|N_i^\pm\cap S|$ and $U=U_i^\pm(S)$.
Since $f_i^\pm$ are discretely concave with curvature $\kappa_i^\pm$, we obtain
\begin{equation}
\label{eq:jensen-plus}
\begin{aligned}
0 &\le f_i^+\bigl(k_i^+(S)\bigr)
 -\mathbb{E}\Bigl[f_i^+\bigl(K_i^+(S)\bigr)\Bigr]\\
&\le \frac{\kappa_i^+}{2}\,
\mathbb{E}\bigl[(U_i^+(S))_2\bigr].
\end{aligned}
\end{equation}
and similarly
\begin{equation}
\label{eq:jensen-minus}
\begin{aligned}
0 &\le f_i^-\bigl(k_i^-(S)\bigr)
 -\mathbb{E}\Bigl[f_i^-\bigl(K_i^-(S)\bigr)\Bigr]\\
&\le \frac{\kappa_i^-}{2}\,
\mathbb{E}\bigl[(U_i^-(S))_2\bigr].
\end{aligned}
\end{equation}
Substituting \eqref{eq:jensen-plus}--\eqref{eq:jensen-minus} into \eqref{eq:Yi-exp} and rearranging, define
\begin{align*}
\Delta_i(S)
&:=\mathbb{E}[Y_i(z_\infty(S))]-\alpha_i q_i(S)\\
&\quad-f_i^+\bigl(k_i^+(S)\bigr)+f_i^-\bigl(k_i^-(S)\bigr)\\
&=-\bigl(f_i^+(k_i^+)-\mathbb{E}[f_i^+(K_i^+)]\bigr)\\
&\quad+\bigl(f_i^-(k_i^-)-\mathbb{E}[f_i^-(K_i^-)]\bigr).
\end{align*}
Taking absolute values and applying the triangle inequality together with \eqref{eq:jensen-plus}--\eqref{eq:jensen-minus} gives
\[
\begin{aligned}
|\Delta_i(S)|
&\le \frac{\kappa_i^+}{2}\,
\mathbb{E}\bigl[(U_i^+(S))_2\bigr]\\
&\quad+\frac{\kappa_i^-}{2}\,
\mathbb{E}\bigl[(U_i^-(S))_2\bigr].
\end{aligned}
\]
This proves the general bound in Theorem~\ref{thm:exposure-compression}.
\end{proof}

\subsection{Network propagation model for moment bounds}
\label{app:diffusion-model}

The remaining statements in Theorem~\ref{thm:exposure-compression} and Lemma~\ref{lem:moment-control} require controlling moments of
\(
U_i^\pm(S)=\sum_{j\in N_i^\pm\setminus S} z_{j,\infty}(S)
\)
under weak propagation.
The moment bounds additionally require control over repeated activation attempts across an unbounded time horizon. We use a standard progressive network propagation model admitting a live-edge representation.

\begin{assumption}[Independent live-edge (IC-type) diffusion: used only for moment bounds]
\label{assump:live-edge}
For each directed edge $e=(u,v)\in E$, let $L_e\sim \mathrm{Bernoulli}(p_{uv})$, independently across $e$.
Given a seed set $S$, define the random directed graph $G_L=(V,E_L)$ with $E_L:=\{e\in E: L_e=1\}$.
The diffusion limit satisfies
\[
z_{v,\infty}(S)=\mathbb{I}\bigl(v\in \mathrm{Reach}_{G_L}(S)\bigr),
\]
i.e., $v$ is active in the limit if and only if it is reachable from $S$ in $G_L$.
\end{assumption}

Assumption~\ref{assump:live-edge} is satisfied by the Independent Cascade model via the classical live-edge coupling \citep{kempe2003maximizing}.
Under this representation, low-probability propagation (Assumption~\ref{assump:lpp}) simply reads $\max_{e\in E}\mathbb{P}(L_e=1)\le \varepsilon$.

For a seed set $S$ and node $v\notin S$, let $\mathcal{P}(S\to v)$ denote the set of \emph{simple} directed paths in $G$ from $S$ to $v$.
Write $|\pi|$ for the number of edges in path $\pi$.
Note that reachability in a finite directed graph always admits a simple path, so it suffices to union bound over $\mathcal{P}(S\to v)$.

\begin{lemma}[Reachability probability under Assumptions~\ref{assump:lpp} and \ref{assump:live-edge}]
\label{lem:reachability}
Assume Assumptions~\ref{assump:lpp} and \ref{assump:live-edge}.
For any seed set $S$ and node $v\notin S$,
\[
\begin{aligned}
\mathbb{P}\bigl(z_{v,\infty}(S)=1\bigr)
&=\mathbb{P}\bigl(v\in \mathrm{Reach}_{G_L}(S)\bigr)\\
&\le\sum_{\pi\in \mathcal{P}(S\to v)} \varepsilon^{|\pi|}\\
&\le\varepsilon\,|\mathcal{P}(S\to v)|.
\end{aligned}
\]
\end{lemma}

\begin{proof}
Under the live-edge model, the event $\{v\in \mathrm{Reach}_{G_L}(S)\}$ implies that there exists at least one simple directed path
$\pi\in\mathcal{P}(S\to v)$ such that all edges on $\pi$ are live.
By the union bound,
\[
\mathbb{P}\bigl(v\in \mathrm{Reach}_{G_L}(S)\bigr)
\le
\sum_{\pi\in\mathcal{P}(S\to v)} \mathbb{P}(\pi \text{ is live}).
\]
Independence of edges yields $\mathbb{P}(\pi \text{ is live})=\prod_{e\in\pi} p_e\le \varepsilon^{|\pi|}$ by Assumption~\ref{assump:lpp}.
Since $v\notin S$, any path has $|\pi|\ge 1$, hence $\varepsilon^{|\pi|}\le \varepsilon$ and the final inequality follows.
\end{proof}

We also need a joint reachability bound for two distinct non-seed nodes.

\begin{lemma}[Joint reachability is second-order]
\label{lem:joint-reachability}
Assume Assumptions~\ref{assump:lpp} and \ref{assump:live-edge}.
Fix $S$ and two distinct nodes $v\neq w$ with $v,w\notin S$.
Then
\[
\begin{aligned}
\mathbb{P}\bigl(z_{v,\infty}(S)=1,\ z_{w,\infty}(S)=1\bigr)
&\le\varepsilon^2\,|\mathcal{P}(S\to v)|\\
&\quad\cdot|\mathcal{P}(S\to w)|.
\end{aligned}
\]
\end{lemma}

\begin{proof}
Under the live-edge model, the event $\{z_{v,\infty}(S)=1,\ z_{w,\infty}(S)=1\}$ implies that there exist
$\pi_v\in\mathcal{P}(S\to v)$ and $\pi_w\in\mathcal{P}(S\to w)$ such that all edges in $\pi_v\cup \pi_w$ are live.
By the union bound,
\[
\begin{aligned}
\mathbb{P}(z_{v,\infty}=1,\ z_{w,\infty}=1)
&\le\sum_{\pi_v\in\mathcal{P}(S\to v)}
\sum_{\pi_w\in\mathcal{P}(S\to w)}\\
&\qquad \mathbb{P}\bigl(\pi_v \text{ and } \pi_w
\text{ are live}\bigr).
\end{aligned}
\]
Since edges are independent,
\[
\begin{aligned}
\mathbb{P}\bigl(\pi_v \text{ and } \pi_w \text{ are live}\bigr)
&=\prod_{e\in \pi_v\cup \pi_w} p_e\\
&\le\varepsilon^{|\pi_v\cup \pi_w|}.
\end{aligned}
\]
Because $v\neq w$ and $v,w\notin S$, the union $\pi_v\cup \pi_w$ contains at least two distinct edges:
each path must contain at least one edge, and if $|\pi_v\cup \pi_w|=1$ then both paths would consist of the same single edge and thus terminate at the same node, contradicting $v\neq w$.
Therefore $|\pi_v\cup \pi_w|\ge 2$ and $\varepsilon^{|\pi_v\cup \pi_w|}\le \varepsilon^2$.
Summing over $(\pi_v,\pi_w)$ yields the claimed bound.
\end{proof}

\subsection{Proof of Lemma~\ref{lem:moment-control}}
\label{app:proof-moment-control}

\begin{proof}[Proof of Lemma~\ref{lem:moment-control}]
Assume Assumptions~\ref{assump:lpp} and \ref{assump:live-edge}.
Fix $S$ with $|S|\le K$ and node $i$.
Write $A_j(S):=|\mathcal{P}(S\to j)|$ for the number of simple directed paths from $S$ to $j$.

\paragraph{First moment.}
By definition,
\[
\begin{aligned}
U_i^\pm(S)&=\sum_{j\in N_i^\pm\setminus S} z_{j,\infty}(S),\\
\mathbb{E}[U_i^\pm(S)]
&=\sum_{j\in N_i^\pm\setminus S}
\mathbb{P}\bigl(z_{j,\infty}(S)=1\bigr).
\end{aligned}
\]
Applying Lemma~\ref{lem:reachability} gives
\[
\begin{aligned}
\mathbb{E}[U_i^\pm(S)]
&\le\sum_{j\in N_i^\pm\setminus S}\varepsilon\,A_j(S)\\
&=\varepsilon\,D_i^\pm(G,S),\\
D_i^\pm(G,S)
&:=\sum_{j\in N_i^\pm\setminus S}A_j(S).
\end{aligned}
\]
Since $G$ is finite and $|S|\le K$, the quantity $D_i^\pm(G,S)$ is finite.
Define
\(
D_i^\pm(G,K):=\max_{S\in\mathcal{S}_K}D_i^\pm(G,S)
\),
which depends only on $(G,K)$ and the exposure neighborhood $N_i^\pm$.
Then $\mathbb{E}[U_i^\pm(S)]\le D_i^\pm(G,K)\,\varepsilon$.

\paragraph{Second factorial moment.}
Recall $(U)_2=U(U-1)=\sum_{j\neq k}\mathbb{I}(j \text{ active})\mathbb{I}(k \text{ active})$.
Thus
\[
\begin{aligned}
\mathbb{E}\bigl[(U_i^\pm(S))_2\bigr]
&=\sum_{\substack{j,k\in N_i^\pm\setminus S\\ j\neq k}}
\mathbb{P}\bigl(z_{j,\infty}(S)=1,\\
&\hspace{5.5em}z_{k,\infty}(S)=1\bigr).
\end{aligned}
\]
Applying Lemma~\ref{lem:joint-reachability} yields
\[
\begin{aligned}
\mathbb{E}\bigl[(U_i^\pm(S))_2\bigr]
&\le\sum_{j\neq k}\varepsilon^2\,A_j(S)A_k(S)\\
&=\varepsilon^2\,C_i^\pm(G,S).
\end{aligned}
\]
where
\(
C_i^\pm(G,S):=\sum_{j\neq k\in N_i^\pm\setminus S}A_j(S)A_k(S)
\).
Define
\(
C_i^\pm(G,K):=\max_{S\in\mathcal{S}_K}C_i^\pm(G,S)
\),
which is finite and depends only on $(G,K)$ and the exposure neighborhood.
Then $\mathbb{E}[(U_i^\pm(S))_2]\le C_i^\pm(G,K)\varepsilon^2$.

\paragraph{Variance.}
Since $K_i^\pm(S)=|N_i^\pm\cap S|+U_i^\pm(S)$ with deterministic $|N_i^\pm\cap S|$, we have $\mathrm{Var}(K_i^\pm(S))=\mathrm{Var}(U_i^\pm(S))$.
Also $U^2=U+(U)_2$, hence
\[
\begin{aligned}
\mathrm{Var}(U)
&\le \mathbb{E}[U^2]
=\mathbb{E}[U]+\mathbb{E}[(U)_2]\\
&\le D_i^\pm(G,K)\varepsilon
+C_i^\pm(G,K)\varepsilon^2
=O(\varepsilon),
\end{aligned}
\]
uniformly over $S\in\mathcal{S}_K$.
\end{proof}

\subsection{Proof of Theorem~\ref{thm:exposure-compression}: second-order moment implication}
\label{app:proof-exposure-secondorder}

\begin{proof}[Proof of Theorem~\ref{thm:exposure-compression}]
The bound is exactly the second-factorial-moment part of Lemma~\ref{lem:moment-control}
(with constants $C_i^\pm(G,K)$ defined there), under Assumption~\ref{assump:live-edge} in addition to Assumption~\ref{assump:lpp}. It yields that the approximation error is $O(\varepsilon^2)$,
which is the statement \eqref{eq:exposure-second-order}.
\end{proof}

\subsection{Proof of Corollary~\ref{cor:welfare-approx} and exact reduction}
\label{app:proof-id-region}
\begin{proof}
Sum the nodewise bounds in Theorem~\ref{thm:exposure-compression} and use the triangle inequality. The uniform constants from Lemma~\ref{lem:moment-control} give
\[
\sup_S|F(S)-\widetilde F(S)|
\le\frac{\varepsilon^2}{2}\sum_i(\kappa_i^+C_i^++\kappa_i^-C_i^-)
=B_{\rm str}\varepsilon^2.
\]
This also proves the enclosure in Corollary~\ref{thm:partial-id}. If the curvatures vanish, each Jensen gap vanishes. If $U_i^\pm\in\{0,1\}$ almost surely, its support lies in one interpolation interval, so the Jensen gap is again zero. Either case proves Corollary~\ref{thm:point-id}.
\end{proof}

\subsection{Response estimator and identifiable coordinates}
\label{app:shape-estimator}
The original component estimator is the convex program in Equations~\eqref{eq:shape-ls}--\eqref{eq:shape-constraints}, with grids and endpoint constraints as specified in Section~\ref{subsec:estimation-fixed}. The statistical theorems concern $\lambda=0$ and a correctly specified conditional mean in this constraint set. The shape restrictions define a closed convex polyhedron \citep{groeneboom2014nonparametric,guntuboyina2018nonparametric}.

\paragraph{Features and interpolation.}
Let $\mathbf b_B(x)$ interpolate the vectors $\mathbf b_B(0)=0$ and $\mathbf b_B(j)=e_j$ for $j=1,\ldots,B$, linearly on each unit interval. Thus $f(x)=\mathbf b_B(x)^\top(f(1),\ldots,f(B))$ for a normalized grid response, and $\|\mathbf b_B(x)\|_2\le1$. In the separate-channel model, use
\begin{align*}
\theta_r&=(\alpha_r,f_r^+(1{:}B_r^+),f_r^-(1{:}B_r^-)),\\
\phi_r(z,u,v)&=(z,\mathbf b_{B_r^+}(u),-\mathbf b_{B_r^-}(v)).
\end{align*}
Rows of $\Phi_{r\ell}$ are the observed features. Evaluation permits $z=q\in[0,1]$, so $\|\phi_r(q,u,v)\|_2\le\sqrt3$ over the whole evaluation domain.

\paragraph{Coincident channels.}
If $N_i^+=N_i^-$ for all nodes in a stratum, then both observed and simulated exposures coincide. Set $h_r=f_r^+-f_r^-$ and use
\[
\theta_r=(\alpha_r,h_r(1{:}B_r)),\qquad
\phi_r(z,k)=(z,\mathbf b_{B_r}(k)).
\]
The induced constraint set consists of net responses representable by the original two constrained curves. It is the linear image of a polyhedron and is therefore closed and convex. Minimizing the original program at $\lambda=0$ induces a minimizer over this set, while predictions and welfare evaluations are unchanged. The net response can be nonmonotone and nonconcave. In fact, every normalized finite-grid response has such a difference representation: choose the slopes and slope drops of a common concave increasing curve large enough that adding the proposed net response preserves both shape conditions. Thus the two component constraints alone impose no restriction on a finite-grid net response; regularization at $\lambda>0$ supplies an additional penalty.

The separate curves are generally unidentifiable in this case: adding $ck$ to both $f_r^+$ and $f_r^-$, for $c\ge0$, preserves the shape restrictions and all fitted values. Theorem~\ref{thm:estimation} therefore bounds net-response error here. For the simulation error below, use the Lipschitz constant of $h_i$ (bounded by $L_i^++L_i^-$) and a single exposure channel. For general separate channels, the theorem requires a nonsingular design in their full coordinates. Even the net coordinates need enough variation to distinguish own activation from exposure effects.

\subsection{Noise, weights, coverage, and effective sample size}
\label{app:design-conditions}
The conditioning variable $\mathcal D$ contains all observed design matrices and the weights. The conditional mean model is
\[
\mathbf Y_{r\ell}=\Phi_{r\ell}\theta_r^\star+\xi_{r\ell},\qquad
\mathbb E[\xi_{r\ell}\mid\mathcal D]=0.
\]
This condition is imposed for the complete design. Merely centering an error conditional on its own node's covariates is weaker. We assume independence of the noise vectors across replications conditional on $\mathcal D$; dependence among nodes within a replication is allowed. This distinction is essential for inference with clustered observations \citep{cameron2015cluster,basse2018correlated}.

For the expectation result, assume conditional covariance bounded by $\sigma_r^2\gamma_r I$. For the high-probability result, impose \eqref{eq:cluster-mgf}, which also implies the covariance bound. Independent conditionally $\sigma_r$-sub-Gaussian coordinates permit $\gamma_r=1$. If only the marginal sub-Gaussian bounds are available, arbitrary dependence among at most $m_r$ nodes permits the conservative choice $\gamma_r=m_r$: H\"older's inequality bounds the joint moment generating function by the product of marginal ones with arguments multiplied by $m_r$. Thus a shared shock can remove the apparent gain from pooling many nodes within one diffusion.

The nonnegative weights satisfy $w_{i\ell}\le w_{\max,r}<\infty$ and are measurable with respect to $\mathcal D$. Unit weights and known design-based weights are examples when the conditional mean condition holds. The information matrix and effective sample size are
\[
Q_r=\sum_\ell\Phi_{r\ell}^\top W_{r\ell}\Phi_{r\ell},\qquad
N_{{\rm eff},r}=\frac{\lambda_{\min}(Q_r)}{\gamma_r w_{\max,r}}.
\]
We require $Q_r\succ0$ in the coordinates used for estimation. Unobserved exposure-bin coordinates give zero columns in the design matrix and make $Q_r$ singular; raw node counts alone cannot ensure coverage or exclude collinearity. Multiplying all weights by the same constant leaves $N_{{\rm eff},r}$ unchanged. Unbounded weights or uncovered grid directions are outside the finite bounds stated here.

All statistical conclusions are conditional on such a covered design. If a design event has probability at least $1-\delta_{\rm des}$, integrating the conditional result gives joint failure probability at most $\delta+\delta_{\rm des}$. A claim about random-design rates therefore additionally requires a lower-tail bound for $Q_r$. This paper does not assume that every finite sample satisfies coverage.

Weak propagation can make coverage harder when informative exposure levels require rare non-seed activations; smaller $\varepsilon$ may then require more logged replications. This is not inevitable: variation in logged seed sets can produce multiple seeded exposures even at $\varepsilon=0$. The moment bounds control non-seed exposure increments and do not imply coverage or its failure.

\subsection{Proof of Theorem~\ref{thm:estimation}}
\label{app:proof-shape}
\begin{proof}
Work in one stratum, suppress its subscript $r$, and put $e=\widehat\theta-\theta^\star$. Optimality of the unregularized constrained fit, with $\theta^\star$ feasible, gives
\[
e^\top Qe\le2e^\top\sum_\ell\Phi_\ell^\top W_\ell\xi_\ell.
\]
Define $v=Q^{-1/2}\sum_\ell\Phi_\ell^\top W_\ell\xi_\ell$. Cauchy--Schwarz gives $\|e\|_Q\le2\|v\|_2$, including the zero-error case. Because $W_\ell^2\preceq w_{\max}W_\ell$ and replications are conditionally independent,
\[
\begin{aligned}
\operatorname{Cov}(v\mid\mathcal D)
&\preceq\sigma^2\gamma Q^{-1/2}\\
&\quad\times\left(\sum_\ell\Phi_\ell^\top W_\ell^2\Phi_\ell\right)Q^{-1/2}\\
&\preceq\sigma^2\gamma w_{\max}I.
\end{aligned}
\]
It follows that
\[
\mathbb E[\|e\|_2^2\mid\mathcal D]
\le\frac{4\,\mathbb E[\|v\|_2^2\mid\mathcal D]}{\lambda_{\min}(Q)}
\le\frac{4\sigma^2\gamma w_{\max}d}{\lambda_{\min}(Q)}.
\]
For the tail bound, applying \eqref{eq:cluster-mgf} to each replication and multiplying their conditional moment generating functions yields
\[
\mathbb E[e^{u^\top v}\mid\mathcal D]
\le e^{\sigma^2\gamma w_{\max}\|u\|_2^2/2}.
\]
The sub-Gaussian quadratic-form inequality of \citet{hsu2012tail}, applied with the identity matrix, gives
\[
\mathbb P\!\left(
\|v\|_2^2>\sigma^2\gamma w_{\max}
[d+2\sqrt{dt}+2t]\mid\mathcal D
\right)\le e^{-t}.
\]
Combine this with $\|e\|_2^2\le4\|v\|_2^2/\lambda_{\min}(Q)$ to obtain \eqref{eq:response-tail}. The same proof applies to the induced net-response constraint set.
\end{proof}

\paragraph{Regularization.}
For $\lambda>0$, the basic inequality has the additional upper bound $N_r\lambda J(\theta^\star)$ on its right-hand side. Consequently,
\[
\|\widehat\theta-\theta^\star\|_2
\le\frac{2\|v\|_2+\sqrt{N_r\lambda J(\theta^\star)}}{\sqrt{\lambda_{\min}(Q)}}.
\]
For a net-response representation, use any feasible true component representation to bound $J(\theta^\star)$. This records the extra term required for a regularized implementation; the displayed main-text rates use $\lambda=0$.

\subsection{Uniform Monte Carlo bounds}
\label{app:uniform-mc}
\label{app:proof-bernstein}
Let $\mathcal J_i$ index the response channels used for node $i$: two signed channels in the general model and one net channel when they coincide. Set $M_{ic}$ to the maximum exposure in channel $c$ and $L_{ic}$ to the Lipschitz constant of its true response. Let
\[
v_{iq}=\sup_{S\in\mathcal S_K}\operatorname{Var}(z_{i,\infty}(S)),\qquad
v_{ic}=\sup_{S\in\mathcal S_K}\operatorname{Var}(K_{ic}(S)).
\]
For a range bound $b\ge0$, define
\[
b_R(v,b,t)=\sqrt{2vt/R}+\frac{2bt}{3R}.
\]
Write $M=|\mathcal S_K|=\sum_{j=0}^{K}\binom nj$ (with $K\le n$). The exact error sums used in the main text are
\begin{align*}
\eta_{\rm MC}(\delta)
&=\sum_i\Bigl[|\alpha_i|b_R(v_{iq},1,t_\delta)\\
&\qquad+\sum_{c\in\mathcal J_i}L_{ic}b_R(v_{ic},M_{ic},t_\delta)\Bigr],\\
\bar\eta_{\rm MC}
&=\sum_i\Bigl[|\alpha_i|b_R(v_{iq},1,t_0)\\
&\qquad+\sum_{c\in\mathcal J_i}L_{ic}b_R(v_{ic},M_{ic},t_0)\Bigr],\\
t_\delta&=\log(12nM/\delta),\qquad t_0=\log(2M)+1.
\end{align*}
The simulation bank is independent of the response noise and produces the correct fixed propagation law conditional on the observed design. Dependence across evaluations of different seed sets is permitted.

\begin{proof}[Proof of Proposition~\ref{prop:mc-concentration} and the uniform bounds]
For fixed $i,S$, Bernstein's inequality for the average of bounded independent simulations gives \eqref{eq:mc-bernstein}. There are at most $3nM$ node/channel/seed-set quantities. Taking $t=t_\delta$ and a union bound makes the probability of any failure at most $\delta/2$.

For an expectation bound, fix one node and channel and let $A$ be the largest absolute simulation error over $S$. A union bound gives
\[
\mathbb P(A>b_R(v,b,\log(2M)+u))\le e^{-u},\qquad u\ge0.
\]
Thus $A$ is stochastically bounded above by $b_R(v,b,\log(2M)+E)$ for a unit-mean exponential variable $E$. The function $b_R$ is concave in its third argument, so $\mathbb E[A]\le b_R(v,b,\log(2M)+1)$. Summing these expected bounds with the displayed coefficients gives $\bar\eta_{\rm MC}$.

Lemma~\ref{lem:moment-control} bounds $v_{ic}=O_{G,K}(\varepsilon)$. For a seed node, the final activation variance is zero; for a non-seed node it is at most its activation probability, bounded by Lemma~\ref{lem:reachability}. Hence $v_{iq}=O_{G,K}(\varepsilon)$ as well. Substitution gives the main-text scaling, with the response coefficients, Lipschitz constants, exposure ranges, and graph-dependent moments retained in its constant.
\end{proof}

\subsection{Proof of Theorem~\ref{thm:fixedS-finite}}
\label{app:proof-fixedS-finite}
For each stratum, let $a_r(S)$ be the sum of its evaluation features at true $(q_i,k_i^\pm)$, and $\widehat a_r(S)$ the corresponding sum at simulated quantities. The interpolation construction gives $\|\widehat a_r(S)\|_2\le\sqrt3m_r$ for every simulation and $S$. In net coordinates the same bound remains valid.

\begin{proof}
Add and subtract the surrogate evaluated with true parameters and simulated propagation quantities:
\begin{equation}
\label{eq:fixedS-decomp}
\begin{aligned}
\widehat F(S)-F(S)
&=\underbrace{\widetilde F(S)-F(S)}_{\text{structural}}\\
&\quad+\underbrace{\sum_r\widehat a_r(S)^\top
  (\widehat\theta_r-\theta_r^\star)}_{\text{response estimation}}\\
&\quad+\underbrace{\sum_r(\widehat a_r(S)-a_r(S))^\top
  \theta_r^\star}_{\text{simulation}}.
\end{aligned}
\end{equation}
Corollary~\ref{cor:welfare-approx} bounds the first term uniformly. The absolute value of the second term is at most $\sqrt3\sum_r m_r\|\widehat\theta_r-\theta_r^\star\|_2$, simultaneously for all $S$ and all simulation banks. By Theorem~\ref{thm:estimation} and a union bound over strata with $t=\log(2R_{\rm str}/\delta)$, it is at most $\eta_{\rm stat}(\delta)$ with conditional probability at least $1-\delta/2$. Its conditional expected supremum is at most $\bar\eta_{\rm stat}$ by Cauchy--Schwarz.

The simulation term is bounded by
\[
\sum_i\left[
|\alpha_i|\,|\widehat q_i(S)-q_i(S)|
+\sum_{c\in\mathcal J_i}L_{ic}|\widehat k_{ic}(S)-k_{ic}(S)|
\right].
\]
Appendix~\ref{app:uniform-mc} bounds its supremum by $\eta_{\rm MC}(\delta)$ with failure probability at most $\delta/2$, and its expected supremum by $\bar\eta_{\rm MC}$. A final union bound proves \eqref{eq:fixedS-rate}; expectation and the triangle inequality prove \eqref{eq:fixedS-bound}.
\end{proof}

\subsection{Consistency and optimization transfer}
\label{app:proof-fixedS-asymp}
\label{app:proof-optimization}
\begin{proof}[Proof of Corollary~\ref{thm:fixedS-asymp}]
The decomposition above also gives
$\mathbb E[\sup_S|\widehat F(S)-\widetilde F(S)|\mid\mathcal D]
\le\bar\eta_{\rm stat}+\bar\eta_{\rm MC}$.
Markov's inequality proves convergence along any covered-design sequence for which this bound vanishes. Add the uniform structural error to obtain convergence to $F$. For random designs, split the probability into the coverage-failure event and its complement; on the complement use the conditional Markov bound truncated at one. If the error bound converges to zero in probability, its truncated version converges to zero in expectation.
\end{proof}

\begin{proof}[Proof of Corollaries~\ref{thm:end2end} and~\ref{cor:opt-rate}]
Let $S^\star$ maximize $F$. On any realization with uniform error at most $\eta$,
\[
F(\widehat S)\ge\widehat F(\widehat S)-\eta
\ge\rho\widehat F(S^\star)-\eta
\ge\rho\bigl(F(S^\star)-\eta\bigr)-\eta
=\rho F(S^\star)-(1+\rho)\eta.
\]
This proves the deterministic transfer. Use the high-probability uniform error from Theorem~\ref{thm:fixedS-finite} for \eqref{eq:policy-highprob}. For \eqref{eq:policy-expectation}, the same deterministic argument uses the random supremum error; take its conditional expectation and apply \eqref{eq:fixedS-bound}. A fixed-set expectation bound alone would not control the data-dependent selected set.
\end{proof}

\section{Implementation Details}
\label{app:implementation}
For each dataset, we preprocess the raw network to obtain a directed interaction graph and the associated metadata used consistently in all experiments.

\subsection{RQ1 (Effectiveness): main evaluation pipeline}
RQ1 evaluates whether CIM selects seed sets that yield higher steady-state causal welfare than influence-maximization baselines under the same budget.
For each dataset and a fixed seed budget $K$, we run each method to obtain a seed set $S$.
CIM follows two stages: (i) learn exposure-response functions from logged diffusion-outcome data with shape constraints (monotonicity and concavity), and (ii) estimate the expected activations and exposures under each seed set by Monte Carlo. Greedy seed selection then maximizes the resulting estimated welfare $\widehat F$.

We evaluate GoodReads, BlogCatalog, Contact-Pri, Pokec, and SD-100.
For GoodReads and Contact-Pri, $N_i^+$ and $N_i^-$ are constructed from the directed in-neighborhood; BlogCatalog and Pokec use the same neighborhood construction on the preprocessed directed interaction graph.
SD-100 comprises 100 nodes and 100 hyperedges. Hyperedge sizes are sampled uniformly from 3 to 8, with members sampled uniformly without replacement. Connecting nodes that share a hyperedge yields 1,201 undirected edges; node degrees have a mean of 24.02, a standard deviation of 8.54, and a range of 3--49.

Baselines construct $S$ under the same $K$: IC-Greedy is classical spread maximization, Degree selects the top-$K$ incidence-degree nodes, Random reports the mean over 10 runs, and A-CauIM/G-CauIM are causal IM baselines.
We then evaluate each selected seed set by simulating diffusion until convergence and computing the resulting steady-state welfare. We use 20 rounds of influence estimation per run and repeat each experiment 10 times.
In addition to welfare, we report wall-clock runtime for seed selection to quantify efficiency.

\paragraph{Implementation efficiency.}
A naive implementation of Algorithm~\ref{alg:cim} requires $\mathcal O(K|V|R)$ diffusion simulations, which is costly for large graphs. Lazy greedy evaluation \citep{minoux1978accelerated} can reuse valid marginal-gain upper bounds in a monotone submodular objective. Candidate pruning restricts the search to the top-$200$ nodes ranked by the degree-weighted exposure score $\deg_i(\widehat f_i^+(\bar k)-\widehat f_i^-(\bar k))$ at a fixed exposure level $\bar k$. A guarantee for this restricted search compares against the retained candidate family; comparing with the full optimum additionally requires a pruning-loss bound. The reported experiments use these operations as implementation heuristics.

\paragraph{Oracle reference.}
The oracle reference seed set $S^\star_{\mathrm{oracle}}$ used in the Oracle Gap metric is computed by running greedy seed selection on the surrogate $\widetilde F(S)$ evaluated with the \emph{ground-truth} response functions $f_i^\pm$ used to generate the logged outcomes (Appendix~\ref{app:synth-data}) and simulation-based estimates of the expected exposures $k_i^\pm(S)$. These ground-truth response functions are used only for the evaluation reference and are never exposed to the compared seed-selection methods, including CIM.

\subsection{RQ2 (Robustness): stress tests on estimation and propagation}
RQ2 tests whether CIM remains reliable when (a) outcome estimation is noisy and (b) the weak-propagation assumption is partially violated.
For outcome noise, we perturb the observed node-level outcomes used in the exposure-response learning stage by injecting additive Gaussian noise $\mathcal{N}(0,\sigma^2)$ with $\sigma\in\{0.0,0.05,0.10,0.20,0.40\}$, then re-fit the response functions and re-run seed selection.
For assumption violations, we increase the propagation strength in the network propagation simulator by scaling edge activation probabilities by a factor $1+\zeta$ for $\zeta\in\{0,0.01,0.02,0.05,0.10,0.20\}$, which amplifies higher-order propagation effects and weakens the validity of the low-propagation approximation.
For each perturbation level, we evaluate the selected seeds via steady-state diffusion simulation and summarize robustness using $\mathrm{Oracle\ Gap}:=1-F_{\mathrm{true}}(\widehat S)/F_{\mathrm{true}}(S^\star_{\mathrm{oracle}})$.
A robust method should exhibit bounded and smooth degradation as noise or propagation strength increases.

\subsection{RQ3 (Sensitivity): budget sensitivity}
RQ3 examines how CIM's performance scales with the seed budget $K$.
We vary $K$ over $\{5,10,15,20\}$ and re-run the full CIM pipeline at each value, then compare the resulting welfare curves against baselines.
On Pokec, CIM's margin over the plotted baselines grows with $K$, while SD-100 exhibits smaller gaps (Figure~\ref{fig:rq3}).

\subsection{Additional empirical results and pseudocode}
\label{app:additional-experiments}

\begin{table}[ht]
\caption{Seed selection efficiency (wall-clock time, single CPU thread, $K=15$).}
\label{tab:time}
\begin{center}
\begin{small}
\begin{sc}
\resizebox{\columnwidth}{!}{
\begin{tabular}{lccccc}
\toprule
Methods & GoodReads & BlogCatalog & Contact-Pri & Pokec & SD-100 \\
\midrule
IC-Greedy & $2.21\,\mathrm{s}$ & $11.41\,\mathrm{s}$ & $25.87\,\mathrm{s}$ & $67.88\,\mathrm{s}$ & $0.94\,\mathrm{s}$ \\
Random & $<1\,\mathrm{ms}$ & $<1\,\mathrm{ms}$ & $<1\,\mathrm{ms}$ & $<1\,\mathrm{ms}$ & $<1\,\mathrm{ms}$ \\
Degree & $<1\,\mathrm{ms}$ & $<1\,\mathrm{ms}$ & $<1\,\mathrm{ms}$ & $4.0\,\mathrm{ms}$ & $<1\,\mathrm{ms}$ \\
A-CauIM & $123.3\,\mathrm{ms}$ & $5.93\,\mathrm{s}$ & $142.9\,\mathrm{ms}$ & $22.01\,\mathrm{s}$ & $38.5\,\mathrm{ms}$ \\
G-CauIM & $2.46\,\mathrm{s}$ & $12.40\,\mathrm{s}$ & $11.57\,\mathrm{s}$ & $414.3\,\mathrm{s}$ & $0.61\,\mathrm{s}$ \\
\textbf{CIM} & $0.18\,\mathrm{s}$ & $7.42\,\mathrm{s}$ & $0.31\,\mathrm{s}$ & $78.4\,\mathrm{s}$ & $51.2\,\mathrm{ms}$ \\
\bottomrule
\end{tabular}
}
\end{sc}
\end{small}
\end{center}
\end{table}

\begin{figure}[ht]
\begin{center}
\begin{subfigure}{0.49\columnwidth}
    \centering
    \includegraphics[width=\linewidth]{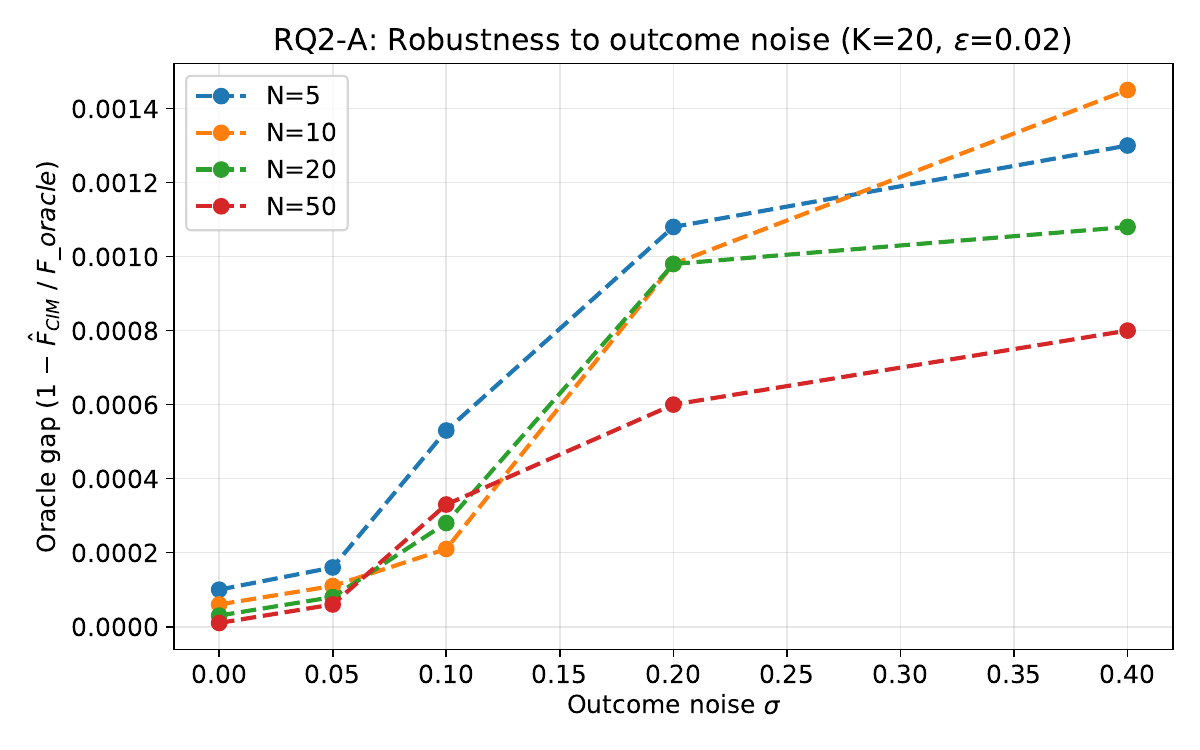}
    \caption{Gap vs. noise $\sigma$}
\end{subfigure}
\hfill
\begin{subfigure}{0.49\columnwidth}
    \centering
    \includegraphics[width=\linewidth]{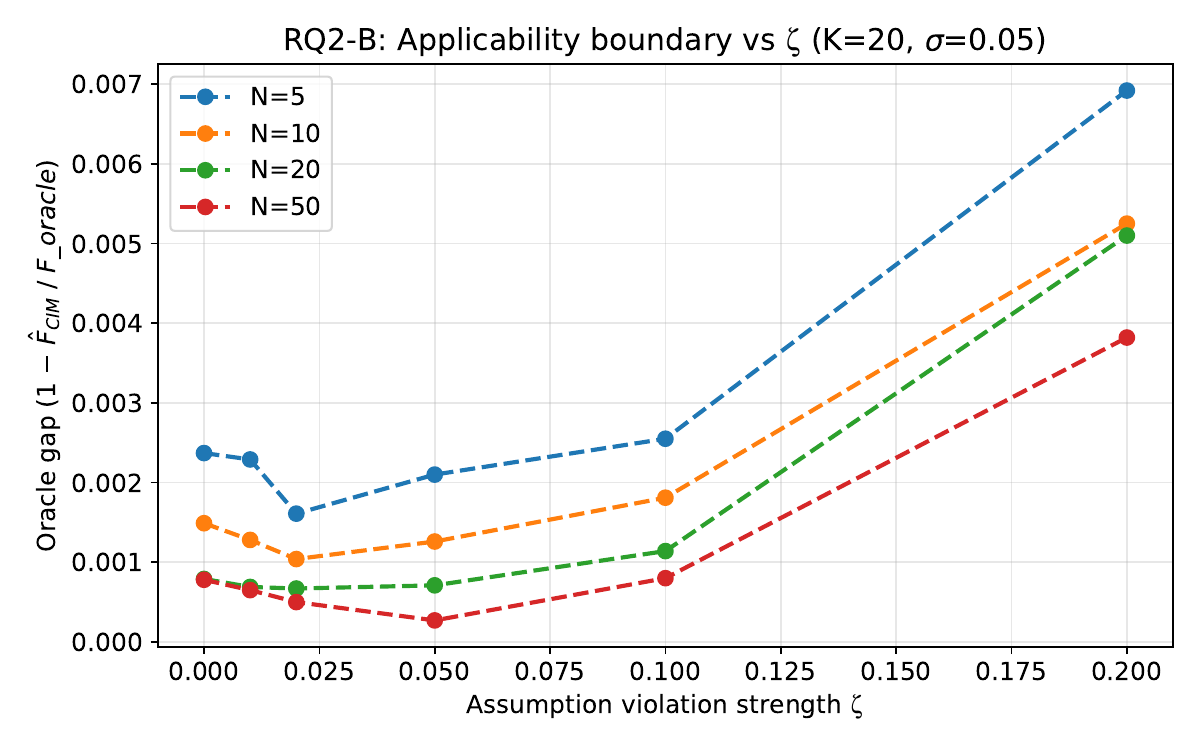}
    \caption{Gap vs. violation $\zeta$}
\end{subfigure}
\caption{RQ2: Robustness analysis on GoodReads ($K=20$). The legend variable $N$ denotes the number of i.i.d.\ logged diffusion-outcome replications used to fit the shape-constrained response functions in Stage~I (Assumption~\ref{assump:replication}); larger $N$ yields more data for response estimation and generally tighter oracle gaps.}
\label{fig:rq2}
\end{center}
\end{figure}

\begin{figure}[ht]
\begin{center}
\begin{subfigure}{0.49\columnwidth}
    \centering
    \includegraphics[width=\linewidth]{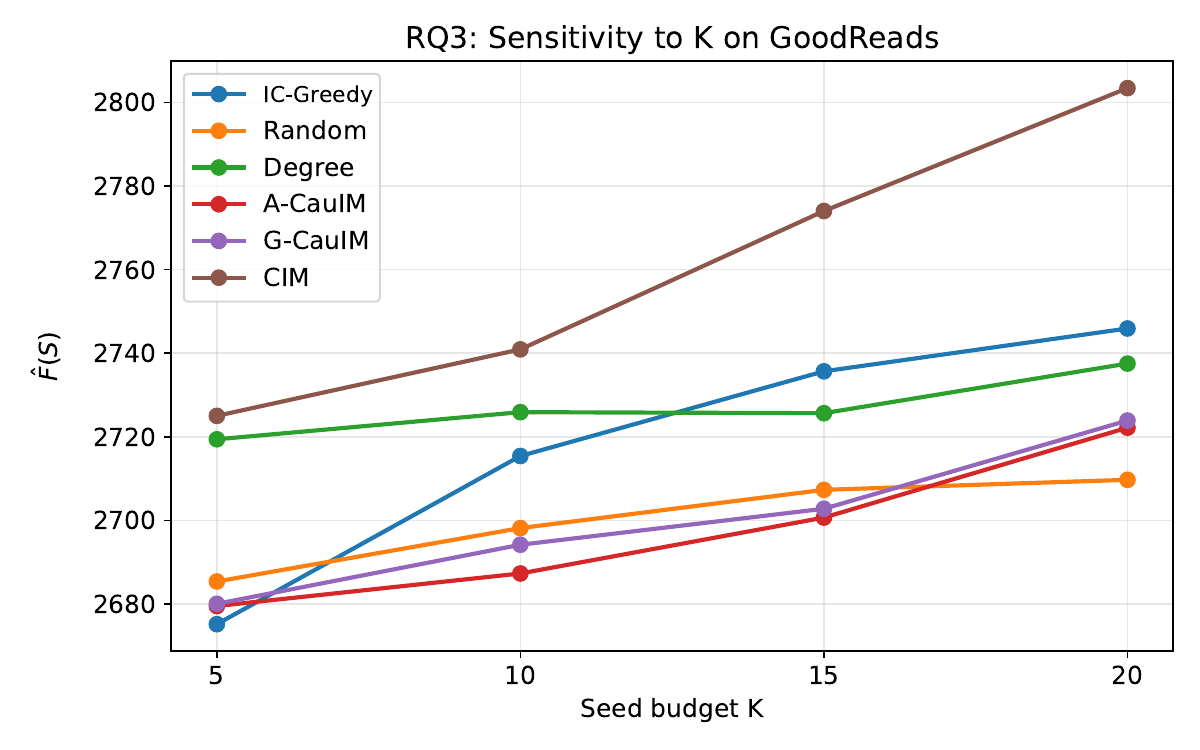}
    \caption{GoodReads}
\end{subfigure}
\hfill
\begin{subfigure}{0.49\columnwidth}
    \centering
    \includegraphics[width=\linewidth]{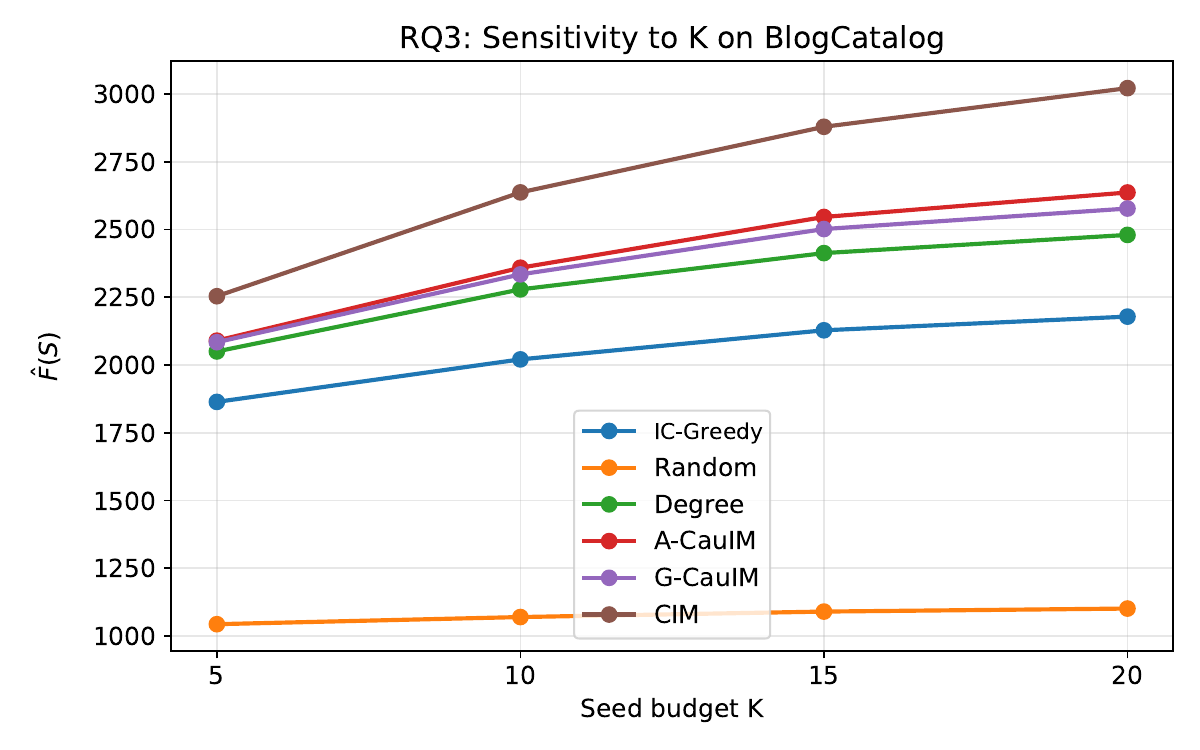}
    \caption{BlogCatalog}
\end{subfigure}
\\[0.5em]
\begin{subfigure}{0.49\columnwidth}
    \centering
    \includegraphics[width=\linewidth]{figures/5.pdf}
    \caption{Pokec}
\end{subfigure}
\hfill
\begin{subfigure}{0.49\columnwidth}
    \centering
    \includegraphics[width=\linewidth]{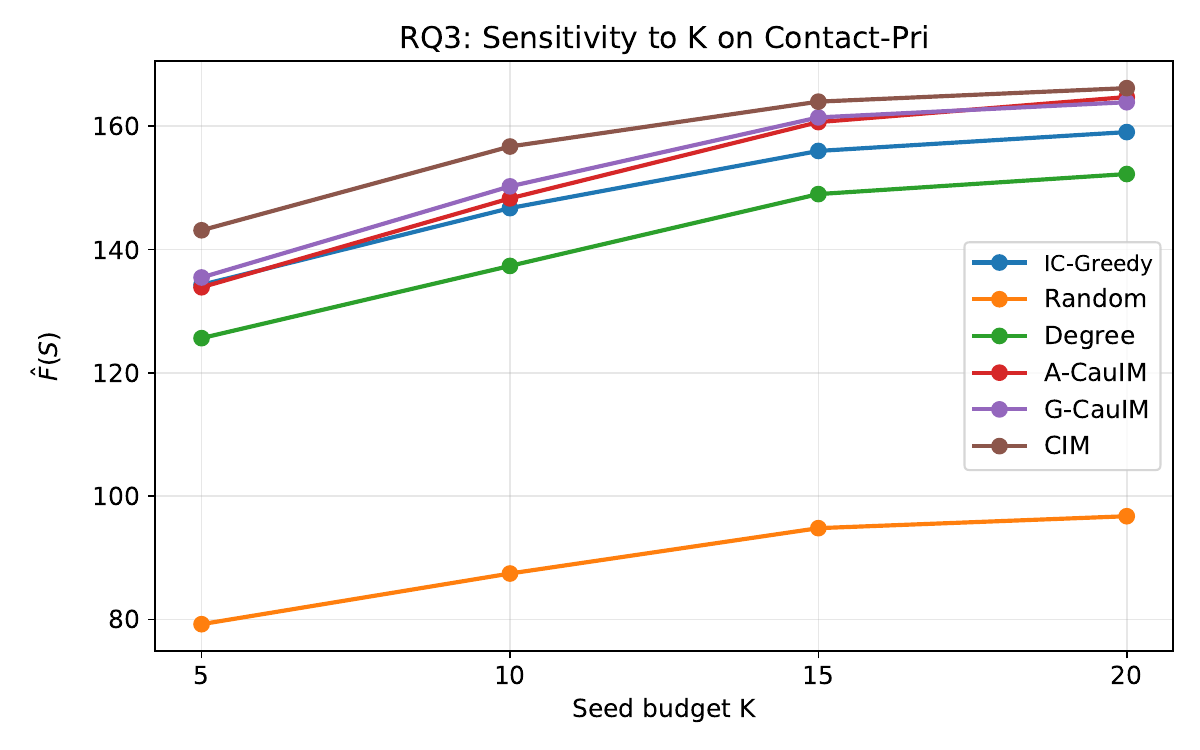}
    \caption{Contact-Pri}
\end{subfigure}
\\[0.5em]
\begin{subfigure}{0.49\columnwidth}
    \centering
    \includegraphics[width=\linewidth]{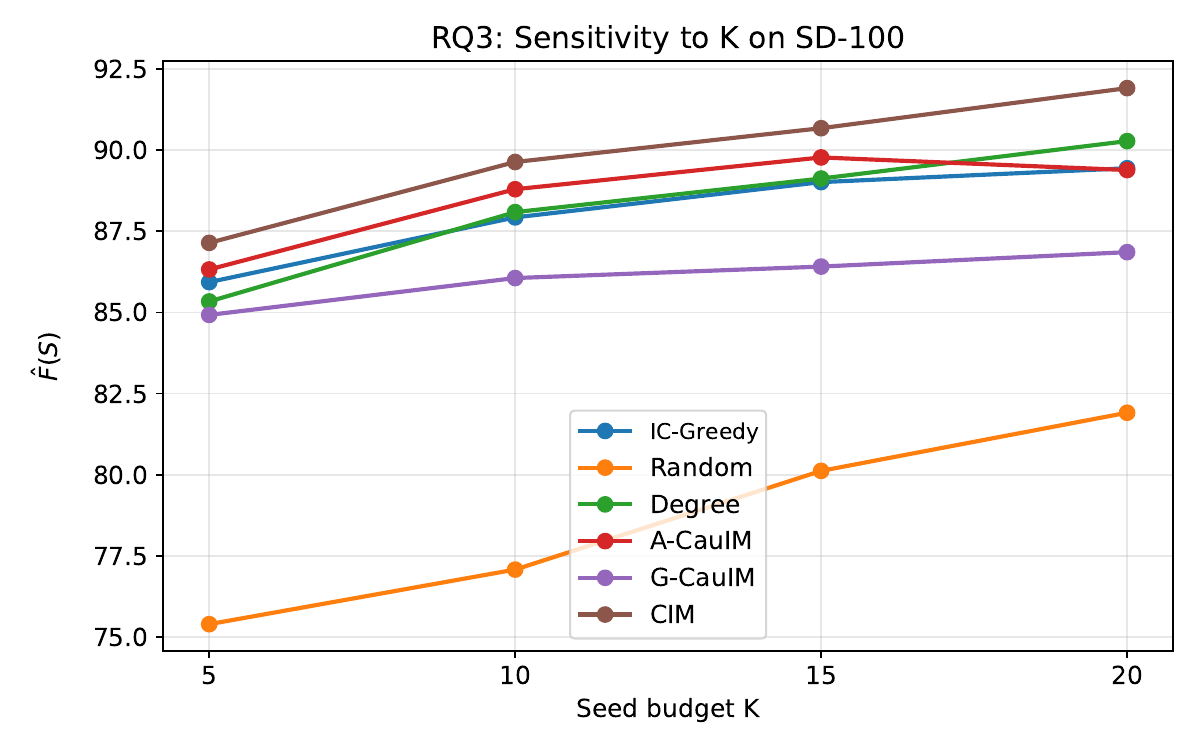}
    \caption{SD-100}
\end{subfigure}

\caption{RQ3: Sensitivity of steady-state welfare $\widehat F(S)$ to seed budget $K$ across all five benchmarks. CIM has the highest welfare at each tested budget on every dataset.}
\label{fig:rq3}
\end{center}
\end{figure}

\section{Synthetic Outcome Generation}
\label{app:synth-data}

The benchmark networks (GoodReads, Contact-Pri, BlogCatalog, Pokec) provide graph topology only; node-level treatments and outcomes are not available natively. We therefore generate logged diffusion-outcome data synthetically following the potential-outcome model of Assumption~\ref{assump:exposure-separable}. SD-100 uses its own simulator and is generated separately. For SD-100, initial treatments are sampled independently from a Bernoulli distribution with probability 0.12. Outcomes are generated by applying a sigmoid transformation to a weighted combination of the treated-neighbor count, Gaussian node heterogeneity, and Gaussian noise.

\paragraph{Edge activation probabilities.}
For each directed edge $(j,i)\in E$, the generator draws $p_{ji}\sim\mathrm{Uniform}(0.01,0.05)$ independently. This bounds individual edge probabilities by $0.05$; the numerical tightness of Theorem~\ref{thm:exposure-compression} also depends on graph-specific path multiplicity and response curvature. RQ2-B scales these probabilities by $1+\zeta$ for $\zeta\in\{0,0.01,0.02,0.05,0.10,0.20\}$. Here $\zeta$ is a relative perturbation level; the theorem's absolute edge bound becomes $\varepsilon=0.05(1+\zeta)$.

\paragraph{Exposure neighborhoods.}
For each node $i$, we set both signed exposure neighborhoods to the directed in-neighborhood: $N_i^+=N_i^-=\{j:(j,i)\in E\}$. Domain-specific signed splits are out of scope for these benchmarks.

\paragraph{Response functions.}
Each node $i$ is assigned to a stratum $r(i)\in\{1,2,3\}$ uniformly at random. Stratum-level parameters are drawn once and shared across nodes within a stratum:
\begin{align*}
\alpha_r &\sim \mathcal{N}(1.0,\,0.1^2),\\
\beta_r^+ &\sim \mathrm{Uniform}(0.5,\,1.5),\quad f_r^+(k)=\beta_r^+\,\log(1+k),\\
\beta_r^- &\sim \mathrm{Uniform}(0.0,\,0.3),\quad f_r^-(k)=\beta_r^-\,\sqrt{k}.
\end{align*}
Both $f_r^+$ and $f_r^-$ are nondecreasing and discretely concave on $\mathbb{Z}_{\geq 0}$, satisfying Assumption~\ref{assump:exposure-separable}. The realized outcome for node $i$ in logged replication $\ell$ is
\[
Y_{i\ell}=\alpha_{r(i)}\,z_{i,\infty,\ell}+f_{r(i)}^+(K_{i\ell}^+)-f_{r(i)}^-(K_{i\ell}^-)+\xi_{i\ell},
\]
with additive Gaussian noise $\xi_{i\ell}\sim\mathcal{N}(0,\sigma^2)$. We use $\sigma=0.05$ in the main experiments; RQ2-A varies $\sigma\in\{0.0,0.05,0.10,0.20,0.40\}$.

\paragraph{Logging policy and dataset size.}
We collect $N=200$ independent logged replications per benchmark. Each seed set contains exactly $K=15$ nodes sampled uniformly without replacement. Thus $\pi_\ell(S)=\binom{n}{K}^{-1}$ for $|S|=K$ and is zero for smaller sets. The IPS identity consequently applies to the supported family of size-$K$ seed sets. Evaluation of other feasible seed sets uses the response model and the diffusion simulator, subject to the coverage conditions in Appendix~\ref{app:design-conditions}. For RQ2, $N\in\{5,10,20,50\}$ denotes the number of logged replications.

\clearpage
\begin{algorithm}[H]
\caption{Two-stage Causal Influence Maximization (CIM)}
\label{alg:cim}
\textbf{Input:}
graph $G=(V,E)$; exposure sets $\{N_i^+,N_i^-\}_{i\in V}$;
data $\mathcal{O}=\{(z_{i\ell},K^+_{i\ell},K^-_{i\ell},Y_{i\ell},X_\ell)\}$;
budget $K$; network propagation simulator or live-edge sampler $\mathsf{Sim}$; number of samples $R$.
\\[2pt]
\textbf{Output:} seed set $\widehat S$ with $|\widehat S|\le K$.
\begin{algorithmic}[1]
\STATE \textbf{(Stage I: Estimate responses)} Fit $\widehat\alpha_r,\widehat f_r^\pm$ by \eqref{eq:shape-ls}--\eqref{eq:shape-constraints} using design-measurable weights. For coincident channels, use the identifiable net response $\widehat h_r$.
\STATE Draw a shared bank of $R$ independent live-edge realizations from $\mathsf{Sim}$ for all candidate evaluations.
\STATE \textbf{Initialize} $\widehat S\leftarrow\varnothing$.
\FOR{$t=1$ \textbf{to} $K$}
    \FOR{\textbf{each} candidate $v\in V\setminus \widehat S$}
        \STATE Estimate $\widehat q_i(\widehat S\cup\{v\})$ and $\widehat k_i^\pm(\widehat S\cup\{v\})$ for all $i$ using the shared simulation bank.
        \STATE Compute the marginal gain $\Delta(v\mid \widehat S):=\widehat F(\widehat S\cup\{v\})-\widehat F(\widehat S)$.
    \ENDFOR
    \STATE Select $v_t\in\arg\max_{v\in V\setminus \widehat S}\Delta(v\mid \widehat S)$ (lazy evaluation can be used when submodularity holds).
    \STATE Update $\widehat S\leftarrow \widehat S\cup\{v_t\}$.
\ENDFOR
\STATE \textbf{return} $\widehat S$.
\end{algorithmic}
\end{algorithm}

\section{Broader Impacts}
\label{app:broader-impacts}

This work studies targeted interventions on networks. Potential positive uses include public-health outreach, information-quality campaigns, and welfare-aware allocation of limited resources, where decision makers need to account for immediate reach, spillovers, saturation, and long-run propagation.

The method is also dual-use. The same targeting machinery could be used for adversarial persuasion, manipulation, discriminatory allocation, or harmful content amplification if the welfare objective is misspecified or maliciously chosen. Even when used with beneficial intent, optimizing aggregate welfare may hide distributional harms across subgroups, and errors in the learned exposure-response functions or in the assumed diffusion model may select seed sets that worsen outcomes for some users.

Responsible deployment should therefore require domain review of the welfare objective, auditing of subgroup-level outcomes, constraints on harmful or manipulative interventions, privacy-preserving handling of logged network data, and monitoring for network drift after deployment. The experiments in this paper are synthetic or based on graph topology only; they should not be interpreted as approval for direct deployment on real users without such safeguards.

\end{document}